\documentclass[10pt,journal,epsfig]{IEEEtran}

\usepackage[dvips]{graphicx}
\usepackage{graphicx}
\usepackage{amssymb}
\usepackage{cite}
\usepackage{subfigure}
\usepackage{amsmath}

\usepackage{subeqnarray}
\usepackage{cases}
\usepackage{color}
\usepackage[centerlast]{caption2}

\begin{document}

\IEEEoverridecommandlockouts
\title{
Improving Wireless Physical Layer Security via Exploiting Co-Channel Interference
}
\author{ 
Lingxiang Li, Athina P. Petropulu, ~\IEEEmembership{Fellow,~IEEE},\\ Zhi Chen, ~\IEEEmembership{Member,~IEEE},
and Jun Fang, ~\IEEEmembership{Member,~IEEE}
\thanks{Lingxiang Li, Zhi Chen, and Jun Fang are with the National Key Laboratory of Science and Technology on Communications,
UESTC, Chengdu 611731, China (e-mails: lingxiang.li@rutgers.edu; \{chenzhi, JunFang\}@uestc.edu.cn).}
\thanks{Athina P. Petropulu is with the Department of Electrical and Computer Engineering, Rutgers--The State University of New Jersey, New Brunswick, NJ 08854 USA (e-mail: athinap@rci.rutgers.edu).}
}

\maketitle

\begin{abstract}
This paper considers a scenario in which a source-destination pair needs to establish
a confidential connection against an external eavesdropper, aided by the interference
generated by another source-destination pair that exchanges public messages.
The goal is to compute the maximum achievable secrecy degrees of freedom (S.D.o.F)
region of a MIMO two-user wiretap network. First, a cooperative secrecy transmission scheme is proposed,
whose feasible set is shown to achieve all S.D.o.F. pairs on the S.D.o.F. region boundary. In this way, the
determination of the S.D.o.F. region is reduced to a problem of maximizing the S.D.o.F. pair
over the proposed transmission scheme. The maximum achievable S.D.o.F. region
boundary points are obtained in closed form, and the construction of the precoding matrices
achieving the maximum S.D.o.F. region boundary is provided.
The obtained analytical expressions clearly show the relation
between the maximum achievable S.D.o.F. region and the number of antennas at each terminal.

\end{abstract}


\section{Introduction}
The area of physical (PHY) layer security has been pioneered by Wyner \cite{Wyner75},
who introduced the wiretap channel and and the notion of secrecy capacity, i.e., the rate at which
the legitimate receiver can correctly decode the source message,
while an unauthorized user, often referred to as eavesdropper, obtains no useful information about the source signal.
For the classical source-destination-eavesdropper Gaussian wiretap channel, the secrecy capacity is zero
when the quality of the legitimate channel is worse than the eavesdropping channel \cite{Leung78}.
One way to achieve non-zero secrecy rates in the latter case is to introduce
one \cite{Swindlehurst11,Lingxiang14,Han11,Ali11,Gan13,zheng151} or more
\cite{Zheng11,Jiangyuan11,LunDong10,Shuangyu13,Kalogerias13,Wang09,Hoon14} external helpers,
who transmit artificial noise, thus acting as jammers to the eavesdropper.
More complex $K$-user interference channels (IFC)
are considered in \cite{Tao10,Ali112,Koyluoglu11,Xie15},
where each user secures its communication from the remaining $K-1$ users by transmitting
jamming signals along with its message signal.

From a system design perspective, introducing non-message carrying artificial noise
into a network is power inefficient and {lowers the overall network throughput}.
In dense multiuser networks there is ubiquitous co-channel interference (CCI), which,
in a cooperative scenario could be designed to effectively act as noise and degrade the eavesdropping channel.
Indeed, there are recent results \cite{Xie15,Xie142,Koyluoglu112,Tung15,Kalantari15,Lv15} on exploiting CCI to enhance secrecy.
\cite{Xie15,Xie142,Koyluoglu112,Tung15} consider the scenario of a $K$-user IFC
in which the users wish to establish secure communication against an eavesdropper.
Specifically, \cite{Xie15,Xie142,Koyluoglu112} consider the single-antenna case
and examine the achievable secrecy degrees of freedom by applying interference alignment techniques.
The work of \cite{Tung15} considers the multi-antenna
case and proposes interference-alignment-based algorithms for the sake of maximizing the achievable secrecy sum rate.
In \cite{Kalantari15,Lv15}, a two-user wiretap interference network is considered, in which only one user needs to establish
a confidential connection against an external eavesdropper, and the secrecy rate is increased by exploiting CCI
due to the nonconfidential connection. \cite{Kalantari15,Lv15} maximize the secrecy transmission rate
of the confidential connection subject to a quality of service constraint for the non-confidential connection.

In this paper, we consider a two-user wiretap interference network as in \cite{Kalantari15,Lv15}, except that,
unlike \cite{Kalantari15,Lv15}, which assume the single input
single-output (SISO) case or multiple-input single-output
(MISO) case,  we address the most general multiple-input multiple-output (MIMO) case,
i.e., the case in which each terminal is equipped with multiple antennas. Out network comprises
a source destination pair exchanging confidential messages, another pair
exchanging public messages, and a passive eavesdropper.
Our goal is to exploit the interference generated by the second source destination pair,
in order to enhance the secrecy rate performance of the network. We should note that, although the eavesdropper is not
interested in the messages of the second pair, for uniformity, we will still refer to the rate of the second pair as secrecy rate.
Since determining the exact maximum achievable secrecy
rate of a helper-assisted wiretap channel, or of an interference channel is a very difficult problem \cite{Swindlehurst11,Lingxiang14,Han11,Ali11,Gan13,zheng151,Zheng11,
Jiangyuan11,LunDong10,Shuangyu13,Kalogerias13,Wang09,Hoon14,Tao10,Ali112},
we consider the high signal to noise ratio (SNR) behavior of the achievable secrecy
rate, i.e., the secrecy degrees of freedom (S.D.o.F.) as an alternative. A similar alternative has also been considered in
\cite{Xie15,Xie142,Koyluoglu112,Nafea13,Nafea14,Nafea15}.
Our main contributions are summarized below.
\begin{enumerate}
\item We propose a cooperative secrecy transmission scheme, in which the
message and interference signals lie in different subspaces at the destination of the confidential connection,
but are aligned along the same subspace at the eavesdropper.
We show that the proposed scheme
can achieve all the boundary points of the S.D.o.F. region (see \emph{Proposition 3}).
In this way, we reduce the determination of each S.D.o.F. region boundary point
to an S.D.o.F. pair maximization problem over our proposed transmission scheme.
\item We determine in closed form the Single-User points,
{\emph{SU}1 and \emph{SU}2 (see eq. (\ref{eq69}) and (\ref{eq510}), respectively) corresponding to when
only one user communicates information},
the {strict S.D.o.F. region boundary} (see eq. (\ref{eq74})), and the ending points of the {strict S.D.o.F. region boundary},
\emph{E}1 and \emph{E}2 (see eq. (\ref{eq521}) and (\ref{eq522}), respectively).
Our analytical results fully describe the dependence of the
S.D.o.F. region of a MIMO two-user wiretap interference channel on the number of antennas.
\item We derive in closed form the general term formulas
for the feasible precoding vector pairs corresponding to the proposed transmission scheme, based on which
we construct precoding matrices achieving S.D.o.F. pairs on the S.D.o.F. region boundary (see Table III).
\end{enumerate}

The corner point of our S.D.o.F. region corresponding to zero S.D.o.F for the nonconfidential connection
has also been studied in \cite{Nafea13,Nafea14,Nafea15}, wherein
the maximum achievable S.D.o.F. of a MIMO wiretap channel with a multi-antenna cooperative jammer has been studied.
Our corner point result is more general because, unlike \cite{Nafea13,Nafea14,Nafea15} it applies to any number of antennas.
It is interesting to note that although we derive the achievable S.D.o.F. from a signal processing point of view,
our corner point result matches the S.D.o.F. result of \cite{Nafea13,Nafea14,Nafea15},
which is derived from an information theoretic point of view.

The idea of signal subspace alignment is also used in  \cite{Agustin11,GouJafar10,YetisGou10,Jiayi14}
in the derivation of the D.o.F. of the $X$ channel and the $K$-user interference channel.
Due to the difference in signal models, the motivation and use of subspace alignment is different.
In \cite{Agustin11,GouJafar10,YetisGou10,Jiayi14}, the authors jointly design the
precoding matrices at the sources, which align multiple interference signals
into a small subspace at each receiver so that the sum dimension of the
interference-free subspaces remaining for the desired signals can be maximized.
In our work, we apply subspace alignment for the sake of degrading the eavesdropping channel
and our goal is to maximize the dimension difference of the interference-free subspaces that the
legitimate receiver and the eavesdropper can see.

The rest of this paper is organized as follows. In
Section II, we introduce a mathematical background, i.e.,
generalized singular value decomposition (GSVD), that
provides the basis for the derivations to follow.
In Section III, we describe the system model for the MIMO two-user wiretap interference channel
and formulate the S.D.o.F. maximization problem. In Section IV, we propose a
secrecy cooperative transmission scheme, and prove that its feasible set is sufficient
to achieve all S.D.o.F. pairs on the  S.D.o.F. region boundary. In Section V,
we determine the maximum achievable S.D.o.F. region boundary,
and uncover its connection to the number of antennas.
In Section VI, we construct the precoding matrices
which achieve the S.D.o.F. pair on the boundary. Numerical results are given in Section VII
and conclusions are drawn in Section VIII.

\textit{Notation:}
$x\sim\mathcal{CN}(0,\Sigma)$ means $x$ is a random variable following a complex circular Gaussian
distribution with mean zero and covariance $\Sigma$; $(a)^+ \triangleq \max(a,0)$;
$\lfloor a\rfloor$ denotes the biggest integer which is less or equal to $a$;
$|a|$ is the absolute value of $a$;
${\bf I}$ represents an identity matrix with appropriate size;
$\mathbb{C}^{N \times M}$ indicates a ${N \times M}$ complex matrix set;
${\bf{A}}^T$, ${\bf{A}}^H$, $\rm{tr}\{\bf{A}\}$,
$\rm{rank}\{\bf{A}\}$, and $|{\bf{A}}|$ stand for the transpose, hermitian transpose, trace,
rank and determinant of the matrix $\bf{A}$, respectively; ${\bf A}(:,j) $
indicates the $j$-th column of $\bf A$ while
and ${\bf A}(:,i:j) $ denotes the columns from $i$ to $j$ of $\bf A$;
${\rm {span}}({\bf A})$ and ${\rm {span}}({\bf A})^\perp$ are the subspace spanned by
the columns of $\bf A$ and its orthogonal complement, respectively;
${\rm {null}}({\bf A})$ denotes the null space of ${\bf A}$;
${\rm {span}}({\bf A})/{\rm {span}}({\bf B})\triangleq\{{{\bf{x}}|{\bf{x}}\in{\rm {span}}({\bf A}),
{\bf{x}} \notin {\rm {span}}({\bf B})}\}$; ${\rm {span}}({\bf A})\cap{\rm {span}}({\bf B})={\bf 0}$
means that ${\rm {span}}({\bf A})$ and ${\rm {span}}({\bf B})$ have no intersections;
${\rm {dim}}\{{\rm {span}}(\bf A)\}$ represents
the number of dimension of the subspace spanned by the columns of $\bf A$;
${\bf \Gamma}({\bf A})$ denotes the orthonormal basis of ${\rm{null}}({\bf A})$;
${\bf A}^\perp$ denotes the orthonormal basis of ${\rm{null}}({\bf A}^H)$.

\newtheorem{proposition}{Proposition}
\newtheorem{theorem}{Theorem}
\newtheorem{corollary}{Corollary}

\section{Mathematical Background}
Given two full rank matrices ${\bf A}\in {{\mathbb C} ^{N \times M}}$ and
${\bf B}\in {{\mathbb C} ^{N \times K}}$.
The GSVD of $({\bf A} ,{\bf B} )$ \cite{Paige81} returns
unitary matrices ${\bf\Psi}_1 \in {{\mathbb C} ^{M\times M}}$, ${\bf\Psi}_2 \in {{\mathbb C} ^{K\times K}}$
and ${\bf\Psi}_0 \in {{\mathbb C} ^{N\times N}}$,
non-negative diagonal matrices ${\bf D}_1\in {{\mathbb C} ^{M\times k}}$ and ${\bf D}_2\in {{\mathbb C} ^{K\times k}}$,
and a matrix ${\bf \Omega} \in {{\mathbb C} ^{k\times k}}$ with ${\rm{rank}}\{{\bf \Omega}\}=k$, such that
\begin{subequations}
\begin{align}
& {\bf A}^H={\bf \Psi}_1{\bf D}_1\left[ {\begin{array}{*{20}{c}}{\bf \Omega}^{-1}&{{{\bf{0}}}}
\end{array}} \right]{\bf \Psi}_0^H, \label{eq401a}\\
& {\bf B}^H={\bf \Psi}_2{\bf D}_2\left[ {\begin{array}{*{20}{c}}{\bf \Omega}^{-1}&{{{\bf{0}}}}
\end{array}} \right]{\bf \Psi}_0^H,
\label{eq401b}
\end{align}
\end{subequations}
with
${{\bf{D}}_1} = \left[ {\begin{array}{*{20}{c}}
{{{\bf{I}}_{r}}}&{\bf{0}}&{\bf{0}}\\
{\bf{0}}&{{{\bf{\Lambda}}_1}}&{\bf{0}}\\
{\bf{0}}&{\bf{0}}&{{{\bf{0}}}}
\end{array}} \right]$,
${{\bf{D}}_2} = \left[ {\begin{array}{*{20}{c}}
{{{\bf{0}}}}&{\bf{0}}&{\bf{0}}\\
{\bf{0}}&{{{\bf{\Lambda}}_2}}&{\bf{0}}\\
{\bf{0}}&{\bf{0}}&{{{\bf{I}}_p}}
\end{array}} \right]$, where the diagonal entries of ${\bf \Lambda}_1\in \mathbb{R}^{s\times s}$
and ${\bf \Lambda}_2\in \mathbb{R}^{s\times s}$ are greater than 0,
and ${\bf{D}}_1^H{{\bf{D}}_1} + {\bf{D}}_2^H{{\bf{D}}_2} = {\bf{I}}$.
It holds that
\begin{subequations}
\begin{align}
k\triangleq &  {\rm {rank}}\{[({\bf A}^H)^T, ({\bf B}^H)^T]^T\}=\min\{M+K,N\}, \label{eq11a}\\
p\triangleq & {\rm {dim}}\{ {\rm {span}}({\bf A})^\perp  \cap {\rm {span}}({\bf B}) \}=k- \min \{M,N\}, \label{eq11b} \\
r \triangleq &  {\rm{dim}}\{{\rm {span}}({\bf A})\cap {\rm {span}}({\bf B})^\perp\}=k- \min \{K,N\}, \label{eq11c}\\
s \triangleq & {\rm {dim}} \{{\rm {span}}({\bf A})\cap {\rm {span}}({\bf B})\}=k-p-r \nonumber \\
 &=(\min \{M,N\}+\min \{K,N\}-N)^+ \label{eq11d}\textrm{.}   %
\end{align}
\end{subequations}

Let ${\bf X}={\bf \Psi}_0\left[ {\begin{array}{*{20}{c}}{\bf \Omega}^{-1}&{{{\bf{0}}}}
\end{array}} \right]^H$ and substitute it into (\ref{eq401a}) and (\ref{eq401b}).
Then, (\ref{eq401a}) and (\ref{eq401b}) can be respectively rewritten as,
\begin{subequations}
\begin{align}
& {\bf A}{\bf \Psi}_1={\bf X}{\bf D}_1^H, \label{eq13a}\\
& {\bf B}{\bf \Psi}_2={\bf X}{\bf D}_2^H
\label{eq13b} \textrm{.}
\end{align}
\end{subequations}
Let ${\bf \Psi}_{11}$, ${\bf \Psi}_{12}$ and ${\bf \Psi}_{13}$ be the
collection of columns $1:r$, $r+1:r+s$, $r+s+1:M$ of ${\bf \Psi}_1$, respectively,
and let ${\bf \Psi}_{21}$, ${\bf \Psi}_{22}$ and ${\bf \Psi}_{23}$ be the
collection of columns $1:K-s-p$, $K-s-p+1:K-p$, $K-p+1:K$ of ${\bf \Psi}_2$, respectively.
In addition, let ${\bf X}_{1}$, ${\bf X}_{2}$ and ${\bf X}_{3}$ be the
collection of columns $1:r$, $r+1:r+s$, $r+s+1:k$ of ${\bf X}$, respectively.
We can rewrite (\ref{eq13a}) and (\ref{eq13b}) as ${\bf A}{\bf \Psi}_{11}={\bf X}_{1}$, ${\bf A}{\bf \Psi}_{12}={\bf X}_{2}{\bf \Lambda}_1$,
${\bf A}{\bf \Psi}_{13}={\bf 0}$; ${\bf B}{\bf \Psi}_{21}={\bf 0}$, ${\bf B}{\bf \Psi}_{22}={\bf X}_{2}{\bf \Lambda}_2$,
${\bf B}{\bf \Psi}_{23}={\bf X}_3$.

In the rest of the paper we will denote the GSVD decomposition in (\ref{eq13a}) and (\ref{eq13b}) as
\begin{align}
{\rm {GSVD}}({\bf A},{\bf B};N,M,K)=
({\bf\Psi}_1,{\bf\Psi}_2,{\bf \Lambda}_1,{\bf \Lambda}_2,{\bf X},k,r,s,p). \nonumber  %
\end{align}

With the GSVD decomposition, one can decompose the union
of ${\rm {span}}({\bf A})$ and ${\rm {span}}({\bf B})$ into three
subspaces, i.e.,
(i) ${\rm {span}}({\bf A})\cap {\rm {span}}({\bf B})^\perp$, which is also the same as ${\rm {span}}({\bf X}_1)$ and has $r$ independent vectors,
(ii) ${\rm {span}}({\bf A})\cap {\rm {span}}({\bf B})$, which is also the same as ${\rm {span}}({\bf X}_2)$ and has $s$ independent vectors,
and (iii) ${\rm {span}}({\bf A})^\perp  \cap {\rm {span}}({\bf B}) $, which is also the same as ${\rm {span}}({\bf X}_3)$ and has $p$ independent vectors.

\begin{proposition}
Consider two full rank matrices ${\bf A}\in {{\mathbb C} ^{N \times M}}$ and ${\bf B }\in {{\mathbb C} ^{N \times K}}$,
and the
${\rm {GSVD}}({\bf A},{\bf B}; N, M, K)=
({\bf\Psi}_1,{\bf\Psi}_2,{\bf \Lambda}_1,{\bf \Lambda}_2,{\bf X},k,r,s,p)$.

(i) ${\bf{A}}{\bf{v}}={\bf{B}}{\bf{w}} \ne {\bf 0} $ holds true
if and only if
\begin{subequations}
\begin{align}
{\bf v}={\bf \Phi}_{1}{\bf y}_{s1}=\left[ {\begin{array}{*{20}{c}}
 {\bf \Psi}_{12}{\bf \Lambda}_1^{-1} &{\bf \Gamma}({\bf A})
\end{array}} \right]
\left[ {\begin{array}{*{20}{c}}
{\bf{y}}_s\\
{\bf{y}}_1
\end{array}} \right], \label{eq80a} \\
{\bf w}={\bf \Phi}_{2}{\bf y}_{s2}=\left[ {\begin{array}{*{20}{c}}
{\bf \Psi}_{22}{\bf \Lambda}_2^{-1} &{\bf \Gamma}({\bf B})
\end{array}} \right]
\left[ {\begin{array}{*{20}{c}}
{\bf{y}}_s\\
{\bf{y}}_2
\end{array}} \right], \label{eq80b}
\end{align}
\end{subequations}
with ${\bf{y}}_s$ being any nonzero vectors, ${\bf{y}}_{s1}$, ${\bf{y}}_{s2}$,
${\bf{y}}_1$ and ${\bf{y}}_2$ being any vectors, with appropriate length.

(ii) The number of linearly independent vectors $\bf v$ satisfying ${\bf{A}}{\bf{v}}={\bf{B}}{\bf{w}} \ne {\bf 0} $
is $s+{\rm {dim}} \{{\rm {null}}({\bf{A}}) \}$.
\end{proposition}
\begin{proof}
See Appendix \ref{appD}.
\end{proof}

\begin{figure}[!t]
\centering
\includegraphics[width=3in]{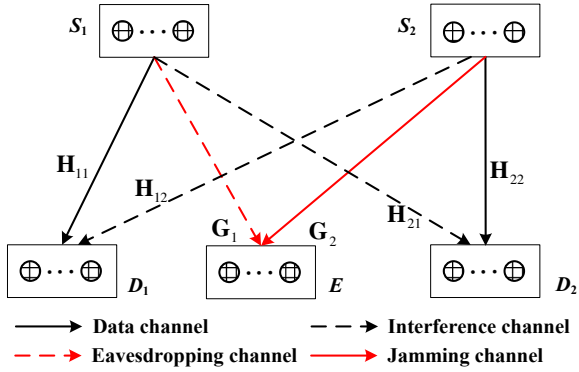}
\DeclareGraphicsExtensions. \caption{A MIMO two-user wiretap interference channel}
\vspace* {-12pt}
\end{figure}
\section{System Model and Problem Statement}
We consider a MIMO interference network which consists of a wiretap channel \emph{S}$_1$-\emph{D}$_1$-\emph{E}
and a point-to-point channel \emph{S}$_2$-\emph{D}$_2$ (see Fig. 1).
In a real setting, the former channel
would correspond to a source-destination pair that needs to maintain secret communications,
while the latter would correspond to a public communication system. While communicating with
its intended destination, \emph{S}$_2$ acts as a jammer to the external passive eavesdropper \emph{E}.
\emph{S}$_1$ and \emph{S}$_2$ are equipped with $N_s^1$, $N_s^2$ antennas, respectively;
\emph{D}$_1$, \emph{D}$_2$ and \emph{E} are equipped with $N_d^1$, $N_d^2$ and
$N_e$ antennas, respectively. Let ${\bf s}_1 \sim \mathcal{CN}(\bf{0},\bf{I})$ and
${\bf s}_2 \sim \mathcal{CN}(\bf{0},\bf{I})$ be the messages transmitted from \emph{S}$_1$ and \emph{S}$_2$, respectively.
Each message is precoded by a matrix before transmission.
The signals received at the legitimate receiver \emph{D}$_i$ can be expressed as
\begin{align}
{{\bf y}_d^i} = {\bf{H}}_{i1}{\bf V}{\bf s}_1 + {{\bf{H}}_{i2}}{{\bf{W}}{\bf s}_2} + {{\bf{n}}_d^i}, i=1, 2,\label{eq1}
\end{align}
while the signal received at the eavesdropper \emph{E} can be expressed as
\begin{align}
{{\bf{y}}_e} = {{\bf{G}}_{1}}{\bf V}{\bf s}_1 + {{\bf{G}}_{2}}{{\bf{W}}{\bf s}_2} + {{\bf{n}}_e}\textrm{.} \label{eq2}
\end{align}
Here, ${\bf{V}}\in\mathbb{C}^{N_s^1 \times K_v}$ and ${\bf{W}}\in\mathbb{C}^{N_s^2 \times K_w}$
are the precoding matrices at \emph{S}$_1$ and \emph{S}$_2$, respectively;
${{\bf{n}}_d^i} \sim \mathcal{CN}(\bf{0},\bf{I}) $ and ${{\bf{n}}_e} \sim \mathcal{CN}(\bf{0},\bf{I})$
represent noise at the $i$th destination \emph{D}$_i$ and the eavesdropper \emph{E}, respectively;
${\bf{H}}_{ij}\in\mathbb{C}^{N_d^i \times N_s^j}$, $i, j \in \{1, 2\}$, denotes the channel matrix from
\emph{S}$_j$ to \emph{D}$_i$;
${\bf{G}}_{j}\in\mathbb{C}^{N_e \times N_s^j}$, $j \in \{1, 2\}$, represents the
channel matrix from \emph{S}$_j$ to \emph{E}.

In this paper, we make the following assumptions:
\begin{enumerate}
\item The messages ${\bf s}_1$ and ${\bf s}_2$ are independent of each other,
and independent of the noise vectors ${\bf n}_d^i$ and ${\bf n}_e$.
\item CCI is treated as noise at each receiver. We assume Gaussian signaling for
\emph{S}$_2$. Thus the MIMO wiretap channel \emph{S}$_1$-\emph{D}$_1$-\emph{E} is
Gaussian. For this case, a Gaussian input signal at \emph{S}$_1$ is the optimal
choice \cite{Liu09, Liu10}.
\item All channel matrices are full rank. Global channel state information (CSI) is available, including the CSI for the eavesdropper.
This is possible in situations in which the eavesdropper is an active member
of the network, and thus its whereabouts and behavior can be monitored.
\end{enumerate}


The achievable secrecy rate for transmitting the message
${\bf s}_1$ and ${\bf s}_2$ are respectively given as \cite{OggierBabak11}
\begin{align}
R_s^1 &=(R_d^1-R_e)^+, \label{eq4}\\
R_s^2 &=R_d^2 \label{eq5} \textrm{.}
\end{align}
where
\begin{subequations}
\begin{align}
& R_d^1 = {\rm {log}}|{\bf I}+({\bf I}+{\bf H}_{12}{\bf Q}_w{\bf H}_{12}^H)^{-1}{\bf H}_{11}{\bf Q}_v{\bf H}_{11}^H|, \label{eq3a}\\
& R_d^2 = {\rm {log}}|{\bf I}+({\bf I}+{\bf H}_{21}{\bf Q}_v{\bf H}_{21}^H)^{-1}{\bf H}_{22}{\bf Q}_w{\bf H}_{22}^H|, \label{eq3b}\\
& R_e = {\rm {log}}|{\bf I}+({\bf I}+{\bf G}_{2}{\bf Q}_w{\bf G}_{2}^H)^{-1}{\bf G}_{1}{\bf Q}_v{\bf G}_{1}^H|,
\label{eq3c}
\end{align}
\end{subequations}
with ${\bf Q}_v\triangleq{\bf {VV}}^H$ and ${\bf Q}_w\triangleq{\bf {WW}}^H$ denoting the transmit
covariance matrices of \emph{S}$_1$ and \emph{S}$_2$, respectively.

The \emph{achievable secrecy rate region} is the set of all secrecy rate pairs, i.e.,
$\mathcal{R}  \buildrel \Delta \over = \mathop  \cup
\limits_{({\bf V}, {\bf W}) \in {\mathcal I}} (R_s^1,R_s^2)$,
where ${\mathcal I}\triangleq \{({\bf V}, {\bf W})|{\rm{tr}} \{ {\bf {VV}}^H\}=P, {\rm{tr}} \{ {\bf {WW}}^H\}= P\}$,
with $P$ denoting the transmit power budget.
Generally, the determination of the outer boundary of $\mathcal{R}$ is a non-convex problem.
Next, we study a simpler problem, namely the \emph{achievable secrecy degrees of freedom region},
defined as
\begin{align}
\mathcal{D}\buildrel \Delta \over = \mathop  \cup \limits_{({\bf V}, {\bf W}) \in {\mathcal I}} ( d_s^1,d_s^2) , \label{eq7}
\end{align}
where $d_s^i$ denotes the high SNR behavior of the achievable secrecy rate, i.e.,
\begin{align}
d_s^i\triangleq \mathop{\lim }\limits_{ P \to \infty } \dfrac{R_s^i}{{\rm log} \ P}, i \in \{1,2\}\label{eq8}.
\end{align}
As shown in Fig. 2, the outer boundary of $\mathcal{D}$ consists of the {strict
S.D.o.F. region boundary} (the part between \emph{E}1 and \emph{E}2 in the graph)
and the {non-strict S.D.o.F. region boundary} (the
vertical part below \emph{E}1 and the horizontal part up to \emph{E}2 of the graph).
The points marked by \emph{SU}1 and \emph{SU}2
correspond to single user S.D.o.F., i.e., when only one user communicates.
For an arbitrary point on the
{strict S.D.o.F. region boundary}, it is impossible to
improve one S.D.o.F., without decreasing the other.
On the other hand, for a point on the {non-strict S.D.o.F. region boundary}, one S.D.o.F. can be further improved
while the other S.D.o.F. remains at the maximum value.

In the following, we will determine the outer boundary of $\mathcal{D}$,
and find its connection to the number of antennas. Towards that goal, we first introduce a cooperative transmission scheme.
Then, by studying that scheme we determine in closed form the outer boundary of $\mathcal{D}$
and also we construct the precoding matrices which achieve the outer boundary of $\mathcal{D}$.

\begin{figure}[!t]
\centering
\includegraphics[width=3in]{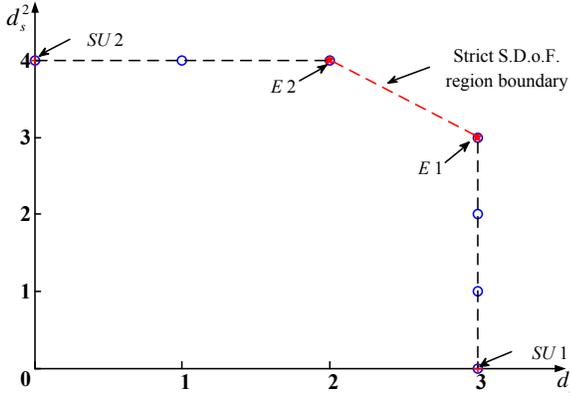}
\DeclareGraphicsExtensions. \caption{Achievable S.D.o.F. region boundary}
\vspace* {-12pt}
\end{figure}

\section{Cooperative Secrecy Transmission Scheme}
\begin{proposition}
For the precoding matrix pair $({\bf V},{\bf W})$,
the achieved S.D.o.F. equals
\begin{subequations}
\begin{align}
&d_s^1({\bf V},{\bf W})= {\rm {rank}}\{{\bf{H}}_{11}{\bf V}\}
-m({\bf V},{\bf W})-n({\bf V},{\bf W}), \label{eq18a}\\
&d_s^2({\bf V},{\bf W})= {\rm{dim}}\{{\rm{span}}({\bf{H}}_{22}{\bf{W}})/{\rm{span}}({\bf{H}}_{21}{\bf{V}})\}, \label{eq18b}
\end{align}
\end{subequations}
in which $m({\bf V},{\bf W})\triangleq {\rm{dim}}\{{\rm{span}}({\bf{G}}_{1}{\bf{V}})/{\rm{span}}({\bf{G}}_{2}{\bf{W}})\} $
and $n({\bf V},{\bf W})\triangleq {\rm{dim}}\{{\rm{span( }}{{\bf{H}}_{12}}{\bf{W}}{\rm{) }} \cap {\rm{span}}({{\bf{H}}_{11}}{\bf{V}})\}$.
\end{proposition}
\begin{proof}
See Appendix \ref{appA}.
\end{proof}

According to \emph{Proposition 2}, the achievable S.D.o.F. of \emph{S}$_1$-\emph{D}$_1$ depends only on the dimension difference
of the interference-free subspaces which \emph{D}$_1$ and \emph{E} can see.
Motivated by this observation, we propose a transmission scheme in which the subspace spanned by the message signal
has no intersection with the subspace spanned by the interference signal at \emph{D}$_1$,
and belongs to the subspace spanned by the interference signal
at \emph{E}. In this way, \emph{D}$_1$ can see an interference-free
message signal, such that $R_d^1$ scales with ${\rm{log}}(P)$, while \emph{E}
can only see a distorted version of the message signal, such that $R_e$
converges to a constant as $P$ approaches to infinity. In other words,
the precoding matrix pairs belongs to the set $\bar {\mathcal I}$, which is defined as follows:
\begin{align}
\bar{\mathcal I} \triangleq \{({\bf V}, {\bf W})|({\bf V}, {\bf W}) \in \bar{\mathcal I}_1 \cap \bar{\mathcal I}_2\cap {\mathcal I}\}, \nonumber
\end{align}
where
\begin{subequations}
\begin{align}
&\bar{\mathcal I}_1 \triangleq \{({\bf V}, {\bf W})|{\rm{span}}({\bf{G}}_{1}{\bf{V}}) \subset {\rm{span}}({\bf{G}}_{2}{\bf{W}})\}, \label{eq10a}\\
&\bar{\mathcal I}_2 \triangleq \{({\bf V}, {\bf W})|{\rm{span}}({\bf{H}}_{11}{\bf{V}}) \cap {\rm{span}}({\bf{H}}_{12}{\bf{W}})={\bf 0}\}.
\label{eq10b}
\end{align}
\end{subequations}

Next, we show that the proposed scheme can achieve all the boundary points of the S.D.o.F. region.
 \begin{proposition}
Let
\begin{align}
\bar{\mathcal{D}}\buildrel \Delta \over = \mathop  \cup \limits_{({\bf V}, {\bf W})
\in \bar{\mathcal I}} ( d_s^1,d_s^2) \label{eq15}\textrm{.}
\end{align}
Then, the outer boundary of $\bar{\mathcal{D}}$ is the same as that of ${\mathcal{D}}$.
\end{proposition}
\begin{proof}
See Appendix \ref{appB}.
\end{proof}

By restricting $({\bf V}, {\bf W})$ to lie in $\bar{\cal I}$,
we exclude a large number of precoding matrix pairs in $\cal I$, which
have no contribution to the outer boundary, and thus reduce
the number of precoding matrices we need to investigate in
determining the outer boundary of the S.D.o.F. region.
It turns out that we can reduce the set even further without
changing the achievable S.D.o.F. region; this is discussed in the following corollary,
where we introduce a new set $\hat {\cal I}$, which is a subset of $\bar{\cal I}$.

\begin{corollary}
Let
\begin{align}
\hat{\mathcal{D}}\buildrel \Delta \over = \mathop
\cup \limits_{({\bf V}, {\bf W}) \in \hat{\mathcal I}} (d_s^1,d_s^2)\textrm{,}
\end{align}
where the set of $\hat{\mathcal I}$ is defined as follows,
\begin{align}
\hat{\mathcal I} \triangleq \{({\bf V}, {\bf W})|{\bf{G}}_{1}{\bf{V}} ={\bf{G}}_{2}{\bf{W}}(:, 1:K_v),
({\bf V}, {\bf W}) \in \bar{\mathcal I}\}. \label{eq00a}
\end{align}
Then, $\hat{\mathcal{D}}=\bar{\mathcal{D}}$.
\end{corollary}
\begin{proof}
See Appendix \ref{appC}.
\end{proof}

\begin{corollary}
For any given precoding matrix pair $({\bf V}, {\bf W}) \in \bar{\mathcal I}$, the
achieved S.D.o.F. over the wiretap channel \emph{S}$_1$-\emph{D}$_1$-\emph{E} is
$d_s^1={\rm rank}({\bf H}_{11}{\bf V})$.
\end{corollary}
\begin{proof}
Since $({\bf V}, {\bf W}) \in \bar{\mathcal I}$, it holds that
${\rm{span}}({\bf{G}}_{1}{\bf{V}}) \subset {\rm{span}}({\bf{G}}_{2}{\bf{W}})$, which indicates
$\mathop{\lim }\limits_{P \to \infty }  \dfrac{R_e}{{\rm{log}}(P)}=0$.
In addition, ${\rm{span}}({\bf{H}}_{11}{\bf{V}}) \cap {\rm{span}}({\bf{H}}_{12}{\bf{W}})={\bf 0}$,
thus $\mathop{\lim }\limits_{P \to \infty }  \dfrac{R_d^1}{{\rm{log}}(P)}={\rm rank}({\bf H}_{11}{\bf V})$.
So,
\begin{align}
d_s^1=\mathop{\lim }\limits_{P \to \infty }  \dfrac{R_d^1}{{\rm{log}}(P)}-
\mathop{\lim }\limits_{P \to \infty } \dfrac{R_e}{{\rm{log}}(P)}={\rm rank}({\bf H}_{11}{\bf V}). \nonumber
\end{align}
This completes the proof.
\end{proof}

\section{Computation of the  S.D.o.F. Boundary}
The key idea for computing the S.D.o.F. boundary
is to maximize the value of $d_s^2$ for a fixed value of $d_s^1$, say $d_s^1=\hat d_s^1$.
Based on \emph{Corollary 1}, in order to determine
the outer boundary of ${\mathcal{D}}$, we only need to focus on the
set $\hat{\mathcal I}$ (see eq. (\ref{eq00a})).
Further, \emph{Corollary 2} shows that for $({\bf V}, {\bf W}) \in \hat{\mathcal I}$
the achieved S.D.o.F. is $d_s^1={\rm {rank}}\{{\bf H}_{11}{\bf V}\}$.
The problem of interest now is to construct precoding matrices
which satisfy $({\bf V}, {\bf W}) \in \hat{\mathcal I}$, $K_v=\hat d_s^1$, and also leave
a maximum dimension interference-free subspace for \emph{D}$_2$.

For ease of exposition, let $({\bf v}, {\bf w})$\footnote{The precoding vector pairs $({\bf v},{\bf w})$
we consider in the construction of $({\bf V}, {\bf{W}})$
are linear independent of each other.} denote the precoding vector pair.
Some observations are in order. First, one can see that when the source message sent by \emph{S}$_1$ lies in the
null space of the eavesdropping channel, even if the pair \emph{S}$_2$-\emph{D}$_2$ communicates, their interference cannot degrade
any further the eavesdropping channel because the eavesdropper already receives nothing; in those cases we may take ${\bf w}={\bf 0}$.
Second, according to \emph{Corollary 2}, for any precoding matrix pairs $({\bf V}, {\bf W}) \in \hat{\mathcal I}$,
the achieved S.D.o.F. $d_s^1={\rm {rank}}\{{\bf H}_{11}{\bf V}\}$. Thus, a greater value of
$d_s^1$ can be achieved by including more linear independent precoding vector pairs in $({\bf V}, {\bf W})$.
Third, the maximum number of linear precoding vector pairs is determined by (\ref{eq10b}),
which requires that
\begin{align}
{\rm{dim}}\{{\rm{span}}({\bf{H}}_{11}{\bf{V}})\}+{\rm{dim}}\{ {\rm{span}}({\bf{H}}_{12}{\bf{W}})\}\le N_d^1. \label{eqdim}
\end{align}
Fourth, the maximum dimension of the interference-free subspace at \emph{D}$_2$ depends on whether \emph{D}$_2$
experiences interference from \emph{S}$_1$.
So, in the following subsections, we will divide the set satisfying ${\bf{G}}_{1}{\bf v} ={\bf{G}}_{2}{\bf w}$ into six subsets,
according to whether the source message from \emph{S}$_1$ lies in the
null space of the eavesdropping channel, whether the source message from \emph{S}$_2$
has interference on \emph{D}$_1$, and whether the source message from \emph{S}$_1$
has interference on \emph{D}$_2$. Accordingly, we characterize the precoding vector pairs in each subset
with the signal dimension triplet $(a,b,c)$, where
$a$ and $b$ denote the number of signal dimensions
we respectively need at \emph{D}$_1$ and \emph{S}$_2$,
and $c$ denotes the signal dimension penalty at \emph{D}$_2$, for obtaining one S.D.o.F.
over the wiretap channel \emph{S}$_1$-\emph{D}$_1$-\emph{E}.
In particular, $a \triangleq {\rm{rank}}\{{\bf{H}}_{11}{\bf{v}}\}+{\rm{rank}}\{{\bf{H}}_{12}{\bf{w}}\}$;
$b \triangleq {\rm{rank}}\{\bf w\}$; $c \triangleq  {\rm{rank}}\{{\bf{H}}_{21}{\bf{v}}\}$. Then,
\begin{enumerate}
\item if the message signal sent by \emph{S}$_1$ spreads within the null space of the eavesdropping channel,
the message signal sent from {\emph S}$_1$ is secure even without the help of \emph{S}$_2$, thus $b=0$, $a=1$;
otherwise, $b=1$.
\item if the message signal sent by \emph{S}$_2$ interferes with \emph{D}$_1$, we need
at least two signal dimensions at \emph{D}$_1$ in order to tell the message signal sent by \emph{S}$_1$
apart from that sent by \emph{S}$_2$, which means that $a=2$; otherwise, $a=1$.
\item if the message signal sent by \emph{S}$_1$ interferes with \emph{D}$_2$,
the signal dimension penalty at \emph{D}$_2$ is one, thus $c=1$; otherwise, $c=0$.
\end{enumerate}
Please refer to Table I for the triplet $(a,b,c)$ of the precoding vector pair from each subset.
Based on this triplet $(a,b,c)$, in this section,
we will analyze the Single-User points \emph{SU}1 and \emph{SU}2,  
the strict S.D.o.F. region boundary, and the ending points of strict S.D.o.F. region boundary
\emph{E}1 and \emph{E}2.

\subsection{Aligned signal subspace decomposition}
In this subsection, we divide the set satisfying ${\bf{G}}_{1}{\bf v} ={\bf{G}}_{2}{\bf w}$
into six subsets, i.e., $Sub_{\rm I}$,..., $Sub_{\rm VI}$, and determine the
number of linear independent precoding vector pairs that should be considered
in each subset, i.e., $d_{\rm I}$,...,$d_{\rm VI}$, respectively.

\begin{table*}[!t]
\renewcommand{\arraystretch}{1.6}
\centering
\caption{The triplet $(a,b,c)$ corresponding to the precoding vector pair from each subset and
the number of linear independent precoding vector pairs that should be considered in each subset}
\begin{tabular}{|c|c|c|}
\hline
\textbf{subsets}& {\textrm{(a,b,c)}}&\textbf{maximum number of linear independent precoding vector pairs $({\bf v}, {\bf w})$} \\
\hline
{$Sub_{\rm I}$}  & $(1,0,0)$ & $d_{\rm I}=(N_s^1-N_e-N_d^2)^+$ \\        
\hline
{$Sub_{\rm II}$}  & $(1,0,1)$ & $d_{\rm II}= \min \{N_d^2, (N_s^1-N_e)^+\}$ \\        
\hline
{$Sub_{\rm III}$}  & $(1,1,0)$ & $d_{\rm III}=(\min \{(N_s^1-N_d^2)^+,N_e\}+\min \{(N_s^2-N_d^1)^+,N_e\}-N_e)^+$ \\        
\hline
{$Sub_{\rm IV}$}  & $(1,1,1)$ & $d_{\rm IV}=(\min \{N_s^1,N_e\}+\min \{(N_s^2-N_d^1)^+,N_e\}-N_e)^+ -d_{\rm III}$ \\        
\hline
{$Sub_{\rm V}$}  & $(2,1,0)$ & $d_{\rm V}=(\min \{(N_s^1-N_d^2)^+,N_e\}+\min \{N_s^2,N_e\}-N_e)^+ -d_{\rm III}$  \\
\hline
{$Sub_{\rm VI}$}  & $(2,1,1)$ & $d_{\rm VI}=(\min \{N_s^1,N_e\}+\min \{N_s^2,N_e\}-N_e)^+-(d_{\rm III}+d_{\rm IV}+d_{\rm V})$    \\
\hline
\end{tabular}
\end{table*}

I) \emph{The message signal sent by { S$_1$} spreads within the null space of the eavesdropping channel, and
does not interfere with {D$_2$}.} That is, the precoding vector pairs in $Sub_{\rm I}$ should satisfy
\begin{subequations}
\begin{align}
&{\bf{G}}_{1}{\bf v}={\bf 0},  \label{eq51a}\\
&{\bf{H}}_{21}{\bf v}={\bf 0}. \label{eq51b}
\end{align}
\end{subequations}
Further, it holds that ${\bf{G}}_{2}{\bf w}={\bf{G}}_{1}{\bf v}={ \bf 0}$.
The case where ${\bf{G}}_{1}{\bf v}={\bf{G}}_{2}{\bf w}=0$ and ${\bf w} \ne 0$ is not considered here,
because even if the pair \emph{S}$_2$-\emph{D}$_2$ communicates, their interference cannot degrade any further the eavesdropping channel.
So we will consider ${\bf w}={\bf 0}$ for simplicity.
Substituting ${\bf v}={\bf \Gamma}({\bf G}_{1}){\bf x}$ into (\ref{eq51b}), with $\bf x$ being any vectors
with appropriate length, we arrive at ${\bf{H}}_{21}{\bf \Gamma}({\bf G}_{1}){\bf x}={\bf 0}$,
which is equivalent to ${\bf x}={\bf \Gamma}({\bf{H}}_{21}{\bf \Gamma}({\bf G}_{1})){\bf y}$, with
$\bf y$ being any vectors with appropriate length. Therefore,
the formula of $\bf v$ in $Sub_{\rm I}$ is
\begin{align}
{\bf v}={\bf \Gamma}({\bf G}_{1}){\bf \Gamma}({\bf{H}}_{21}{\bf \Gamma}({\bf G}_{1})){\bf z},  \label{eq52a1}
\end{align}
with $\bf z$ being any nonzero vectors with appropriate length.
In addition, since all the channel matrices are assumed to be full rank, it holds that
\begin{align}
d_{\rm I} \le {\rm {dim}}\{{\rm {null}}({{\bf H}_{21} {\bf \Gamma}}({\bf{G}}_{1}))\}= (N_s^1-N_e-N_d^2)^+. \label{eq52}
\end{align}

II) \emph{The message signal sent by { S$_1$} spreads within the null space of the eavesdropping channel,
but does interfere with { D$_2$}.} That is, the vectors in $Sub_{\rm II}$ should satisfy
\begin{subequations}
\begin{align}
&{\bf{G}}_{1}{\bf v}={\bf 0}, \label{eq53a} \\
&{\bf{H}}_{21}{\bf v}\ne {\bf 0}. \label{eq53b}
\end{align}
\end{subequations}
Here again, we will consider ${\bf w}={\bf 0}$ for simplicity.
On combining (\ref{eq51a})-(\ref{eq51b}) with (\ref{eq53a})-(\ref{eq53b}), it holds that
\begin{align}
Sub_{\rm I} \cup Sub_{\rm II}=\{({\bf v},{\bf w})|{\bf{G}}_{1}{\bf v}={\bf 0}, {\bf w}={\bf 0}\}.
\label{eq54a3}
\end{align}
So, the linear independent vectors we can choose from $Sub_{\rm I}$ and $Sub_{\rm II}$
should be no greater than ${\rm {dim}} \{{\rm {null}}({\bf{G}}_{1})\}$. That is,
\begin{align}
d_{\rm II} +d_{\rm I}\le  (N_s^1-N_e)^+. \label{eq54}   
\end{align}

III) \emph{The message signal sent by { S$_1$} does not spread within the null space of the eavesdropping channel.
The message signals sent by { S$_1$} and { S$_2$} do not interfere with { D$_2$} and { D$_1$}, respectively.}
That is, the precoding vector pairs in $Sub_{\rm III}$ should satisfy
\begin{subequations}
\begin{align}
&{\bf{H}}_{12}{\bf w}={\bf 0}, \label{eq55a} \\
&{\bf{H}}_{21}{\bf v}= {\bf 0}, \label{eq55b}\\
&{\bf{G}}_{1}{\bf v} ={\bf{G}}_{2}{\bf w} \ne {\bf 0}. \label{eq55c}
\end{align}
\end{subequations}
Substituting ${\bf v}={\bf \Gamma}({\bf H}_{21}){\bf x}$ and ${\bf w}={\bf \Gamma}({\bf H}_{12}){\bf y}$
into (\ref{eq55c}), we arrive at
\begin{align}
{\bf{G}}_{1}{\bf \Gamma}({\bf H}_{21}){\bf x} = {\bf{G}}_{2}{\bf \Gamma}({\bf H}_{12}){\bf y} \ne {\bf 0}. \label{eq56}
\end{align}
Consider the decomposition
\begin{align}
&{\rm {GSVD}}({\bf{G}}_{1}{\bf \Gamma}({\bf H}_{21}), {\bf{G}}_{2}{\bf \Gamma}({\bf H}_{12}); N_e, \hat N_s^1, \hat N_s^2) \nonumber\\
&\quad \quad =(\hat {\bf\Psi}_1,\hat{\bf\Psi}_2,\hat{\bf \Lambda}_1,\hat{\bf \Lambda}_2,\hat{\bf X},\hat k,\hat r,\hat s,\hat p), \nonumber
\end{align}
where $\hat N_s^1 \triangleq (N_s^1-N_d^2)^+$ and $\hat N_s^2 \triangleq(N_s^2-N_d^1)^+$.
Applying \emph{Proposition 1}, we
can obtain the number of linearly independent vectors $\bf v$ satisfying (\ref{eq56}), i.e.,
\begin{align}
\hat d_{\rm III} \triangleq \hat s+{\rm {dim}} \{{\rm {null}}({\bf{G}}_{1}{\bf \Gamma}({\bf H}_{21})) \}. \nonumber 
\end{align}
Since ${\rm {null}}({\bf{G}}_{1}{\bf \Gamma}({\bf H}_{21}))={\rm {null}}({{\bf H}_{21} {\bf \Gamma}}({\bf{G}}_{1})) $,
the basis of ${\rm {null}}({\bf{G}}_{1}{\bf \Gamma}({\bf H}_{21}))$ also spans the solution space of $\bf v$ in $Sub_{\rm I}$. Thus,
\begin{align}
d_{\rm III} +d_{\rm I}\le \hat d_{\rm III}=\hat s+ (N_s^1-N_e-N_d^2)^+,  \label{eq57a2}
\end{align}

IV) \emph{The message signal sent by {S$_1$} does not spread within the null space of the eavesdropping channel. 
The message signal sent by {S$_2$} does not interfere with {D$_1$}, but the message signal
sent by {S$_1$} interferes with {D$_2$}.}
That is, the precoding vector pairs in $Sub_{\rm IV}$ should satisfy
\begin{subequations}
\begin{align}
&{\bf{H}}_{12}{\bf w}={\bf 0},  \label{eq58a}\\
&{\bf{H}}_{21}{\bf v} \ne {\bf 0}, \label{eq58b} \\
&{\bf{G}}_{1}{\bf v} ={\bf{G}}_{2}{\bf w} \ne {\bf 0}. \label{eq58c}
\end{align}
\end{subequations}
Substituting ${\bf w}={\bf \Gamma}({\bf{H}}_{12}){\bf y}$
into (\ref{eq58c}), we get
\begin{align}
{\bf{G}}_{1}{\bf v} = {\bf{G}}_{2}{\bf \Gamma}({\bf{H}}_{12}){\bf y} \ne {\bf 0}. \label{eq59}
\end{align}
Consider the decomposition
\begin{align}
&{\rm {GSVD}}({\bf{G}}_{1}, {\bf{G}}_{2}{\bf \Gamma}({\bf H}_{12}); N_e, N_s^1,\hat N_s^2) \nonumber \\
&\quad \quad =(\bar {\bf\Psi}_1,\bar{\bf\Psi}_2,\bar{\bf \Lambda}_1,\bar{\bf \Lambda}_2,\bar{\bf X},\bar k,\bar r,\bar s,\bar p). \nonumber
\end{align}
Applying \emph{Proposition 1} we
can obtain the number of linearly independent vectors $\bf v$ satisfying (\ref{eq59}), i.e.,
\begin{align}
&\hat d_{\rm IV} \triangleq \bar s + {\rm {dim}} \{{\rm {null}}({\bf{G}}_{1}) \}. \nonumber 
\end{align}
On combining (\ref{eq55a})-(\ref{eq55c}) with (\ref{eq58a})-(\ref{eq58c}), it holds that
\begin{align}
Sub_{\rm III} \cup Sub_{\rm IV}=\{({\bf v}, {\bf w})|{\bf{H}}_{12}{\bf w}={\bf 0},
{\bf{G}}_{1}{\bf v} ={\bf{G}}_{2}{\bf w} \ne {\bf 0}\} \nonumber
\end{align}
In addition, the basis of ${\rm {null}}({\bf{G}}_{1})$ also spans the solution space
of $\bf v$ in $Sub_{\rm I}\cup Sub_{\rm II}$. Therefore,
\begin{align}
d_{\rm IV}+d_{\rm III}+ d_{\rm II}+d_{\rm I} \le \hat d_{\rm IV} =\bar s+ (N_s^1-N_e)^+. \label{eq61}
\end{align}

V) \emph{The message signal sent by {S$_1$} does not spread within the null space of the eavesdropping channel.
The message signal sent by {S$_2$} interferes with {D$_1$}, but the message signal
sent by {S$_1$} does not interfere with {D$_2$}.}
That is, the precoding vector pairs in $Sub_{\rm V}$ should satisfy
\begin{subequations}
\begin{align}
&{\bf{H}}_{12}{\bf w}\ne {\bf 0},  \label{eq61a}\\
&{\bf{H}}_{21}{\bf v} = {\bf 0},  \label{eq61b}\\
&{\bf{G}}_{1}{\bf v} ={\bf{G}}_{2}{\bf w} \ne {\bf 0}. \label{eq61c}
\end{align}
\end{subequations}
Substituting  ${\bf v}={\bf \Gamma}({\bf{H}}_{21}){\bf x}$ into (\ref{eq61c}), we obtain
\begin{align}
{\bf{G}}_{1}{\bf \Gamma}({\bf{H}}_{21}){\bf x}= {\bf{G}}_{2}{\bf w} \ne {\bf 0}. \label{eq62}
\end{align}
Consider the decomposition
\begin{align}
&{\rm {GSVD}}({\bf{G}}_{1}{\bf \Gamma}({\bf H}_{21}), {\bf{G}}_{2}; N_e, \hat N_s^1, N_s^2) \nonumber \\
&\quad \quad =(\breve {\bf\Psi}_1,\breve{\bf\Psi}_2,\breve{\bf \Lambda}_1,\breve{\bf \Lambda}_2,\breve{\bf X},\breve k,\breve r,\breve s,\breve p). \nonumber
\end{align}
Applying \emph{Proposition 1}, we can obtain the number of linearly independent vectors $\bf v$ satisfying
(\ref{eq62}), i.e.,
\begin{align}
\hat d_{\rm V} \triangleq  \breve s+{\rm {dim}} \{{\rm {null}}({\bf{G}}_{1}{\bf \Gamma}({\bf H}_{21})) \}.\nonumber  
\end{align}
On combining (\ref{eq55a})-(\ref{eq55c}) with (\ref{eq61a})-(\ref{eq61c}), it holds that
\begin{align}
Sub_{\rm III} \cup Sub_{\rm V}=\{({\bf v}, {\bf w})|{\bf{H}}_{21}{\bf v} = {\bf 0},{\bf{G}}_{1}{\bf v} ={\bf{G}}_{2}{\bf w} \ne {\bf 0}\} \nonumber
\end{align}
In addition, the basis of ${\rm {null}}({\bf{G}}_{1}{\bf \Gamma}({\bf H}_{21}))$ also spans the solution space of $\bf v$ in $Sub_{\rm I}$. Therefore,
\begin{align}
d_{\rm V}+d_{\rm III}+d_{\rm I} \le \hat d_{\rm V}=\breve s+ (N_s^1-N_e-N_d^2)^+. \label{eq64a1}
\end{align}

VI) \emph{The message signal sent by {S$_1$} does not spread within the null space of the eavesdropping channel.
The message signals sent by {S$_2$} and {S$_1$} interfere with {D$_1$} and {D$_2$}, respectively.}
That is, the precoding vector pairs in $Sub_{\rm VI}$ should satisfy
\begin{subequations}
\begin{align}
&{\bf{H}}_{12}{\bf w} \ne {\bf 0},  \label{eq65a}\\
&{\bf{H}}_{21}{\bf v} \ne {\bf 0}, \label{eq65b}\\
&{\bf{G}}_{1}{\bf v} ={\bf{G}}_{2}{\bf w} \ne {\bf 0}. \label{eq65c}
\end{align}
\end{subequations}
Consider the decomposition
\begin{align}
&{\rm {GSVD}}({\bf{G}}_{1}, {\bf{G}}_{2}; N_e, N_s^1, N_s^2) =(\tilde {\bf\Psi}_1,\tilde{\bf\Psi}_2,\tilde{\bf \Lambda}_1,\tilde{\bf \Lambda}_2,\tilde{\bf X},\tilde k,\tilde r,\tilde s,\tilde p). \nonumber
\end{align}
According to \emph{Proposition 1}, we can obtain the number of linearly independent vectors $\bf v$ satisfying (\ref{eq65c}), i.e.,
\begin{align}
&d_{s} \triangleq \tilde s+{\rm {dim}} \{{\rm {null}}({\bf{G}}_{1})\}. \nonumber 
\end{align}
On combining (\ref{eq65a})-(\ref{eq65c}) with (\ref{eq55a})-(\ref{eq55c}),
(\ref{eq58a})-(\ref{eq58c}) and (\ref{eq61a})-(\ref{eq61c}), it holds that
$Sub_{\rm III} \cup Sub_{\rm IV}\cup Sub_{\rm V} \cup Sub_{\rm VI}=\{({\bf v}, {\bf w})|{\bf{G}}_{1}{\bf v} ={\bf{G}}_{2}{\bf w} \ne {\bf 0}\} $.
In addition, the basis of ${\rm {null}}({\bf{G}}_{1})$ also spans the solution space of $\bf v$ in $Sub_{\rm I}\cup Sub_{\rm II}$. Thus,
\begin{align}
d_{\rm VI}+d_{\rm V} +d_{\rm IV}+d_{\rm III}+d_{\rm II}+d_{\rm I} \le d_s. \label{eq65}
\end{align}

We should note that with all three variables smaller than the corresponding variables of other triplets,
the precoding vector pair from $Sub_{ \rm I}$ has the potential to achieve a greater S.D.o.F. than the others,
and so it has the highest priority in the construction of $({\bf V}, {\bf W})$.
Similarly, the precoding vector pair from $Sub_{\rm IV}$ has lower priority than that
one from $Sub_{\rm I} \cup Sub_{\rm II} \cup Sub_{\rm III}$; the precoding vector pair from
$Sub_{\rm V}$ has lower priority than that one from $Sub_{\rm I} \cup Sub_{\rm III}$;
and the precoding vector pair from $Sub_{ \rm VI}$ has the lowest priority.
Therefore, all the equalities in (\ref{eq52}), (\ref{eq54}), (\ref{eq57a2}), (\ref{eq61}), (\ref{eq64a1}) and (\ref{eq65}) hold true.
As a conclusion, the number of linear independent precoding vector pairs that should be considered in each subset is given in Table I.

Correspondingly, in what follows, we give the formulas of $\bf v $ and $\bf w$ we consider in each subset.
Combining the formula of $\bf v $ in $Sub_{\rm I}$, i.e., (\ref{eq52a1}),
and that one in $Sub_{\rm I} \cup Sub_{\rm II}$, i.e., (\ref{eq54a3}),
we obtain the one in $Sub_{\rm II}$, i.e.,
\begin{align}
{\bf v}={\bf \Gamma}({\bf G}_{1}){\bf \Gamma}^\perp({\bf{H}}_{21}{\bf \Gamma}({\bf G}_{1})){\bf z}+
{\bf \Gamma}({\bf G}_{1}){\bf \Gamma}({\bf{H}}_{21}{\bf \Gamma}({\bf G}_{1})){\bf y}.  \nonumber
\end{align}
with $\bf z$ being any nonzero vectors with appropriate length.
Since we want linear independent precoding vectors, the beamforming direction
already considered in the set with higher priority, e.g., $Sub_{\rm I}$, should not be under consideration
in other subsets. Thus, the formula of $\bf v$ in $Sub_{\rm II}$ is
\begin{align}
{\bf v}={\bf \Gamma}({\bf G}_{1}){\bf \Gamma}^\perp({\bf{H}}_{21}{\bf \Gamma}({\bf G}_{1})){\bf z}.  \label{eq54a2}
\end{align}
Similarly, the formulas of $\bf v$ and $\bf w$ in $Sub_{\rm III}$ are, respectively,
\begin{align}
{\bf v}=\hat{\bf \Psi}_{12}\hat{\bf \Lambda}_1^{-1}{\bf z},
{\bf w}=\hat{\bf \Psi}_{22}\hat{\bf \Lambda}_2^{-1}{\bf z}. \label{eq57a3}
\end{align}
The formulas of $\bf v$ and $\bf w$ in $Sub_{\rm IV}$ are, respectively,
\begin{align}
{\bf v}=\bar{\bf \Psi}_{12}\bar{\bf \Lambda}_1^{-1}{\bf z},
{\bf w}=\bar{\bf \Psi}_{22}\bar{\bf \Lambda}_2^{-1}{\bf z}. \label{eq61a2}
\end{align}
The formulas of $\bf v$ and $\bf w$ in $Sub_{\rm V}$ are, respectively,
\begin{align}
{\bf v}=\breve{\bf \Psi}_{12}\breve{\bf \Lambda}_1^{-1}{\bf z},
{\bf w}=\breve{\bf \Psi}_{22}\breve{\bf \Lambda}_2^{-1}{\bf z}. \label{eq64a2}
\end{align}
And the formulas of $\bf v$ and $\bf w$ in $Sub_{\rm VI}$ are, respectively,
\begin{align}
{\bf v}=\tilde{\bf \Psi}_{12}\tilde{\bf \Lambda}_1^{-1}{\bf z},
{\bf w}=\tilde{\bf \Psi}_{22}\tilde{\bf \Lambda}_2^{-1}{\bf z}. \label{eq67a2}
\end{align}
We should note that since ${\bf H}_{21}$ is independent of the channels ${\bf{G}}_{1}$, ${\bf{G}}_{2}$
and ${\bf H}_{12}$, for precoding vector pairs in (\ref{eq61a2})
${\bf H}_{21}{\bf v} \ne 0$ holds true with probability one. Similar argument also applies in the derivation of
the formulas of $\bf v$ and $\bf w$ in $Sub_{\rm V}$ and $Sub_{\rm VI}$.

\subsection{Single-User points {SU}\emph{1($\bar d_s^1$, 0)} and {SU}\emph{2(0, $\bar d_s^2$)}}
A single-user point corresponds to a scenario in which only one source-destination communicates.
Let $\bar d_s^1$ and $\bar d_s^2$ denote the maximum achievable value of $d_s^1$ and $d_s^2$, respectively.

\subsubsection{The single-user point {SU}\emph{1($\bar d_s^1$, 0)}}
In this case, the pair \emph{S}$_2$-\emph{D}$_2$ does not communicate, but \emph{S}$_2$ still transmits, acting as a cooperative jammer
targeting at degrading the eavesdropping channel. In this case, the system model reduces to a wiretap channel with a
cooperative jammer.
Based on \emph{Corollary 1} and \emph{Corollary 2}, we see that
our problem for maximizing $d_s^1$ is including as more precoding vector pairs as possible in $({\bf V},{\bf W})$.
In Table I, we divide the set which satisfies ${\bf G}_{1}{\bf v}={\bf G}_{2}{\bf w}$ into six subsets.
Due to the requirement in (\ref{eqdim}), it holds that
more precoding vector pairs can be included in $({\bf V},{\bf W})$ by
choosing precoding vector pairs from the subsets with smaller $a$.
For example, $a=1$ for $Sub_{\rm IV}$ while $a=2$ for $Sub_{\rm VI}$.
We can select at most $N_d^1$ precoding vector pairs from
$Sub_{\rm IV}$, in which $a=1$, while we can select only $\lfloor N_d^1/2\rfloor$ precoding vector pairs from $Sub_{\rm VI}$,
in which $a=2$. In addition, since the achieved S.D.o.F. is $d_s^1={\rm {rank}}\{{\bf H}_{11}{\bf V}\}$,
a greater value of $d_s^1$ can be achieved with precoding vector pairs from $Sub_{\rm IV}$.
Therefore, in the construction of $({\bf V},{\bf W})$, the precoding vector pairs from the first four subsets have the
same priority, and the precoding vector pairs from the last two subsets have the same priority.
Moreover, a precoding vector pair from the first four subsets has higher priority than that one from
the last two subsets.
If $N_d^1 \le d_{\rm I}+d_{\rm II}+d_{\rm III}+d_{\rm IV}$, we just select
$N_d^1$ precoding vector pairs from $Sub_{\rm I}\cup Sub_{\rm II} \cup Sub_{\rm III} \cup Sub_{\rm IV}$; otherwise,
we first select all the precoding vector pairs in $Sub_{\rm I}\cup Sub_{\rm II} \cup Sub_{\rm III} \cup Sub_{\rm IV}$,
and then we pick $\lfloor \dfrac{N_d^1-(d_{\rm I}+d_{\rm II}+d_{\rm III}+d_{\rm IV})}{2}\rfloor$
precoding vector pairs from $Sub_{\rm V}\cup Sub_{\rm VI}$.

\emph{Example 1:} Consider the case $(N_s^1, N_d^1, N_e)=(6,3,6)$, $(N_s^2, N_d^2)=(6,6)$.
Based on Table I, the maximum number of linear
independent precoding vector pairs in each subset is $d_{\rm I}=0$, $d_{\rm II}=0$, $d_{\rm III}=0$,
$d_{\rm IV}=3$, $d_{\rm V}=0$, $d_{\rm VI}=3$. Since $N_d^1=d_{\rm I}+d_{\rm II}+d_{\rm III}+d_{\rm IV}$, we first select three precoding
vector pairs in $Sub_{\rm IV}$. We cannot pick any more precoding vector pairs without violating (\ref{eqdim}) since in that case the the remaining
signal dimension at \emph{D}$_1$ is $N_d^1- d_{\rm IV}=0$.
Concluding, we can select a total of 3 precoding vector pairs, and based on \emph{Corollary 2},
$\bar d_s^1=3$.

\emph{Example 2:} Consider the case $(N_s^1, N_d^1, N_e)=(6,5,5)$, $(N_s^2, N_d^2)=(6,4)$.
Based on Table I we get that $d_{\rm I}=0$, $d_{\rm II}=1$, $d_{\rm III}=0$,
$d_{\rm IV}=1$, $d_{\rm V}=2$, $d_{\rm VI}=2$. Since $N_d^1> d_{\rm I}+d_{\rm II}+d_{\rm III}+d_{\rm IV}$, we first select all the precoding
vector pairs in $Sub_{\rm II}$ and $Sub_{\rm IV}$, i.e., $({\bf v}_1,{\bf w}_1)$, $({\bf v}_2,{\bf w}_2)$, with ${\bf H}_{12}{\bf w}_1=0$
and ${\bf H}_{12}{\bf w}_2=0$. From the remaining sets
$Sub_{\rm V}$ and $Sub_{\rm VI}$, we can at most pick one pair, i.e., $({\bf v}_3,{\bf w}_3)$. For either $Sub_{\rm V}$ or $Sub_{\rm VI}$,
it holds that ${\bf H}_{12}{\bf w}_3 \ne 0$. Thus, for ${\bf V}=[{\bf v}_1\ {\bf v}_2\ {\bf v}_3]$ and ${\bf W}=[{\bf w}_1\ {\bf w}_2\ {\bf w}_3]$
it holds that ${\rm dim}\{{\rm span}({\bf H}_{11}{\bf V})\}+{\rm dim}\{{\rm span}({\bf H}_{12}{\bf W})\}=3+1=4$.
If we picked another pair, (\ref{eqdim}) would be violated.
Concluding, we can select a total of 3 precoding vector pairs, and based on \emph{Corollary 2}, $\bar d_s^1=3$.

Summarizing, the maximum achievable value $\bar d_s^1$, i.e.,
\begin{align}
\bar d_s^1=\min \{d_{a=1}+d_{a=2}^\star,N_d^1\}, \label{eq69}
\end{align}
where $d_{a=1}=d_{\rm I}+d_{\rm II}+d_{\rm III}+d_{\rm IV}$, and
\begin{align}
&d_{a=2}^\star=\min \{d_{\rm V}+d_{\rm VI}, \lfloor   {(N_d^1-d_{a=1})^+}/{2} \rfloor\}\nonumber \textrm{.}
\end{align}

\emph{Remark 1:}
To gain more insight into $\bar d_s^1$, we give Table II which shows the dependence of
$\bar d_s^1$ on the number of antennas.

\begin{table*}[!t]
\renewcommand{\arraystretch}{1.6}
\centering
\caption{Summary of the closed-form results on $\bar d_s^1$}
\begin{tabular}{|c|c|}
\hline
\textbf{Inequalities on the number of antennas at terminals} & $\bar d_s^1$  \\
\hline
$N_s^1 \ge N_e+N_d^1$  & \\ \cline{1-1}
$N_s^2 \ge N_e+N_d^1$ & $\min \{N_s^1,N_d^1\}$  \\ \cline{1-1}
$2N_d^1+N_e-N_s^2 \le N_s^1 < N_e+N_d^1$ & \\
$N_d^1 < N_s^2 < N_e +N_d^1$ & \\
\hline
$N_d^1+N_e-N_s^2 < N_s^1 <2N_d^1+N_e-N_s^2$ & $N_s^1+N_s^2-(N_d^1+N_e)+\min \{s,\lfloor \frac {2N_d^1+N_e-N_s^1-N_s^2}{2} \rfloor\}$\\
$N_d^1 < N_s^2 < N_e +N_d^1$ & $s=\min\{N_d^1+N_e-N_s^2,N_e\}+\min\{N_s^2,N_e\}-N_e$\\
\hline
$N_e < N_s^1 < N_e +N_d^1$, $N_s^2 \le N_d^1$  & $N_s^1-N_e+\min \{s,\lfloor \frac {N_d^1+N_e-N_s^1}{2} \rfloor\}$, $s=\min\{N_s^2,N_e\}$ \\
\hline
$N_s^1 \le N_d^1+N_e-N_s^2$, $N_d^1 < N_s^2 < N_e +N_d^1 $ & $\min \{s,\lfloor \frac {N_d^1}{2} \rfloor\}$\\ \cline{1-1}
$N_s^1 \le N_e$, $N_s^2 \le N_d^1$ & $s=\min\{N_s^1,N_e\}+\min\{N_s^2,N_e\}-\min\{N_s^1+N_s^2,N_e\}$ \\
\hline
\end{tabular}
\end{table*}

\subsubsection{The single-user point of {SU}\emph{2(0, $\bar d_s^2$)}}
In this case, the wiretap channel \emph{S}$_1$-\emph{D}$_1$-\emph{E} does not work.
For a point-to-point MIMO user, the maximum achievable degrees of freedom equals $\min\{N_s^2, N_d^2\}$. That is,
\begin{align}
\bar d_s^2=\min\{N_s^2, N_d^2\}. \label{eq510}
\end{align}

\subsection{Computation of the strict  S.D.o.F. region boundary}
The key idea for computing the strict  S.D.o.F. boundary
is to maximize the value of $d_s^2$ for a fixed value of $d_s^1$.

Assume that ${\bf V}$ consists of $\hat d_s^1$ columns, among which
$z$ columns come from a subset for which the message
signal sent by \emph{S}$_1$ interferes with \emph{D}$_2$.
Then, \emph{D}$_2$ can at most see a $(N_d^2-z)^+$-dimension interference-free subspace.
Thus,
\begin{align}
\hat d_s^2(z) \le (N_d^2-z)^+ .  \label{eqds21}
\end{align}
In addition, it holds that $\hat d_s^1+{\rm{dim}}\{ {\rm{span}}({\bf{H}}_{12}{\bf{W}})\}\le N_d^1$ due to (\ref{eqdim}). So,
\begin{align}
{\rm rank}\{{\bf W}\}\le (\max \{N_s^2,N_d^1\}-\hat d_s^1)^+. \label{eqds22}
\end{align}
Combining (\ref{eq510}), (\ref{eqds21}) and (\ref{eqds22}), we get the maximum achievable value of $d_s^2$, i.e.,
\begin{align}
\hat d_s^2(z)=\min \{N_s^2,(\max \{N_s^2,N_d^1\}-\hat d_s^1)^+, (N_d^2-z)^+\}. \label{eq72}
\end{align}
Thus, in order to maximize the value of $d_s^2$, we only need to minimize the value of $z$.

According to Table I, the minimum value of $z$ without the constraint $d_s^1=\hat d_s^1$ equals
$(\hat d_s^1-(d_{\rm V}+d_{\rm I}+d_{\rm III}))^+$.
Due to the constraint $d_s^1=\hat d_s^1$ and the fact that $a=2$ in $Sub_{\rm V}$, we have limitations on
the number of pairs that can be selected from $Sub_{\rm V}$.
For example, consider the case $d_{\rm I}+d_{\rm III}=2$, $d_{\rm V}=2$, $N_d^1=3$ and $\hat d_s^1=3$. The minimum
value of $z$ without the constraint $d_s^1=\hat d_s^1=3$ equals $0$, in which case we need at least choose
one pair from $Sub_{\rm V}$. Noting that (\ref{eqdim}) should be satisfied for $({\bf V}, {\bf W})\in \hat{\cal I}$ and $a=2$ in $Sub_{\rm V}$,
if we have picked one pair from $Sub_{\rm V}$, we can then at most pick one more pair from the first four subsets. Thus,
the maximum achievable value of $d_s^1$ equals 2, which violates the constraint $d_s^1=3$.
Due to the constraint $d_s^1=3$ and the fact that $a=2$ in $Sub_{\rm V}$, we cannot select
any pairs from $Sub_{\rm V}$, and so the minimum value of $z$ equals to 1.

Let $x$ and $y$ denote the number of columns which
come from the first four subsets and the last two subsets, respectively. The maximum allowable value of $y$
under the constraint of $d_s^1=\hat d_s^1$ is
\begin{subequations}
\begin{align}
y_{\max} \triangleq &\max_{x, y} \quad  y \nonumber \\
{\rm s.t.} \quad & x+y=\hat d_s^1,  \label{eq71a}\\
& x+2y \le N_d^1, \label{eq71b}\\
& 0 \le x \le d_{\rm I}+d_{\rm II}+d_{\rm III}+d_{\rm IV}, \label{eq71c}\\
& 0 \le y \le d_{\rm V}+d_{\rm VI}. \label{eq71d}
\end{align}
\end{subequations}
Substituting $x=\hat d_s^1-y$ into (\ref{eq71b}), we arrive at $y \le N_d^1-\hat d_s^1$, which
combined with (\ref{eq71c}) and (\ref{eq71d}) gives
\begin{align}
y_{\max} = \min \{N_d^1-\hat d_s^1, d_{\rm V}+d_{\rm VI}, \hat d_s^1\}.  \label{eq100}
\end{align}
Thus, we can select at most $\min \{y_{\max}, d_{\rm V}\}$ precoding vector pairs from $Sub_{\rm V}$.
Therefore, the minimum value of $z$ is,
\begin{align}
z_{\min}(\hat d_s^1)=(\hat d_s^1-(\min \{y_{\max}, d_{\rm V}\}+d_{\rm I}+d_{\rm III}))^+.  \label{eq73}
\end{align}
Substituting (\ref{eq73}) into (\ref{eq72}), we obtain the maximum value of $d_s^2$, i.e.,
\begin{align}
\hat d_s^2=\min \{N_s^2,(\max \{N_s^2,N_d^1\}-\hat d_s^1)^+, (N_d^2-z_{\min}(\hat d_s^1))^+\}. \label{eq74}
\end{align}

\emph{Remark 2}: For any given values of $d_s^1$, we can derive a maximum achievable value of $d_s^2$
based on (\ref{eq74}). Finally, the strict S.D.o.F.
region boundary can be computed based on the following iteration:
\begin{enumerate}
\item Initialize $\hat d_s^1=\bar d_s^1$;
\item Compute $\hat d_s^2$ with (\ref{eq74});
\item Compare $\hat d_s^2$ with $\bar d_s^2$. If $\hat d_s^2<\bar d_s^2$, let $\hat d_s^1=\hat d_s^1-1$ and go to 2);
otherwise, stop and output all the pairs $(\hat d_s^1,\hat d_s^2)$.
\end{enumerate}

\emph{Example 3:}
Let us revisit \emph{Example 2}, for which we obtained
$\bar d_s^1=3$ and $\bar d_s^2=4$, respectively. Initialize $\hat d_s^1$ with $\bar d_s^1=3$.
Substituting $\hat d_s^1=3$ into (\ref{eq74}), we obtain $\hat d_s^2=3$.
Since $\hat d_s^2<\bar d_s^2$, we continue the iteration. Letting $\hat d_s^1=2$ and substituting it into
(\ref{eq74}), we obtain $\hat d_s^2=4$, which equals $\bar d_s^2$.
So, we stop the iteration and output all the S.D.o.F. pairs on the strict  S.D.o.F. region boundary, i.e.,
$(\hat d_s^1,\hat d_s^2)=(3,3)$ and $(\hat d_s^1,\hat d_s^2)=(2,4)$.

\subsection{Ending points of strict  S.D.o.F. region boundary {E}\emph{1($\bar d_s^1$, $\underline{d}_s^2$)}
and {E}\emph{2($ \underline{d}_s^1$, $\bar d_s^2$)}}
As shown in Fig. 2, \emph{E}1 and \emph{E}2 denote the ending points of the {strict S.D.o.F. region boundary}.
In particular, $\underline{d}_s^2$ denotes the maximum achievable value of ${d}_s^2$ under the constraint ${d}_s^1=\bar{d}_s^1$,
and $\underline{d}_s^1$ denotes the maximum achievable value of ${d}_s^1$ under the constraint ${d}_s^2=\bar{d}_s^2$.

\emph{1) The ending point {E}\emph{1($\bar d_s^1$, $\underline{d}_s^2$)}}.
According to (\ref{eq69}), we obtain $\bar d_s^1$ which denotes the maximum achievable value of $d_s^1$.
Substituting $\hat d_s^1=\bar d_s^1$ into (\ref{eq100})-(\ref{eq74}), we arrive at
\begin{align}
\underline{d}_s^2=\min \{N_s^2,(\max \{N_s^2,N_d^1\}-\bar d_s^1)^+, (N_d^2-z_{\min}(\bar d_s^1))^+\}.  \label{eq521}
\end{align}

\emph{2) The ending point {E}\emph{2($ \underline{d}_s^1$, $\bar d_s^2$)}}.
According to the previous analysis
on the single-user point of \emph{SU}{2(0, $\bar d_s^2$)}, we obtain $\bar d_s^2=\min\{N_s^2, N_d^2\}$,
which, combined with (\ref{eq72}), gives
\begin{subequations}
\begin{align}
&\min\{N_s^2, N_d^2\}  \le \max \{N_s^2,N_d^1\}-\underline{d}_s^1,  \label{eq172} \\
&\min\{N_s^2, N_d^2\}  \le N_d^2-z . \label{eq173}
\end{align}
\end{subequations}
In the following, we consider two distinct cases.

(i) For the case of $N_s^2 > N_d^2$, (\ref{eq172}) becomes
\begin{align}
\underline{d}_s^1 \le  \max \{N_s^2,N_d^1\} - N_d^2. \label{eq75}
\end{align}
Besides, (\ref{eq173}) indicates that $z=0$, and thus all of the signal steams
sent by \emph{S}$_1$ should not interfere with \emph{D}$_2$.
That is, $Sub_{\rm II}$, $Sub_{\rm IV}$ and $Sub_{\rm VI}$ are not under consideration.
Applying (\ref{eq69}), we obtain
\begin{align}
\underline{d}_s^1 \le \min\{ d_{\rm I}+d_{\rm III}+\beta^\star, N_d^1\}, \label{eq76}
\end{align}
where$\beta^\star=\min \{d_{\rm V}, \lfloor   {(N_d^1-d_{\rm I}-d_{\rm III})^+}/{2} \rfloor\}$.
Combining (\ref{eq75}) and (\ref{eq76}), we arrive at
\begin{align}
\underline{d}_s^1 = \min \{d_{\rm I}+d_{\rm III}+\beta^\star,\max \{N_s^2,N_d^1\} - N_d^2, N_d^1\}.
\end{align}

(ii) For the case of $N_s^2 \le N_d^2$, (\ref{eq172}) becomes
\begin{align}
&\underline{d}_s^1 \le  \max \{N_s^2,N_d^1\} - N_s^2, \label{eq174}
\end{align}
which indicates that $\underline{d}_s^1=0$ when $N_s^2 \ge N_d^1$.
So, in the following, we only consider the case of $N_s^2 < N_d^1$, where it holds that
$d_{\rm III}=d_{\rm IV}=0$.
In addition, (\ref{eq173}) indicates that
$z  \le N_d^2- N_s^2.  $ 
Therefore, $\xi = \min\{d_{\rm VI}, (N_d^2- N_s^2-d_{\rm II})^+\} +d_{\rm V}$, 
where $\xi$ denotes the maximum number of precoding vector pairs that can be chosen from
$Sub_{\rm V}$ and $Sub_{\rm VI}$.
Applying (\ref{eq69}), we get
\begin{align}
\underline{d}_s^1 \le \min\{ d_{\rm I}+\hat d_{\rm II}+\xi^\star, N_d^1\}, \label{eq177}
\end{align}
where $\hat d_{\rm II}=\min \{ N_d^2- N_s^2 ,d_{\rm II}\}$, and
\begin{subequations}
\begin{align}
&\xi^\star=\min \{\xi, \lfloor   {(N_d^1-d_{\rm I}-\hat d_{\rm II})^+}/{2} \rfloor\}\nonumber \textrm{.}
\end{align}
\end{subequations}
Combining (\ref{eq174}) and (\ref{eq177}), we arrive at
\begin{align}
\underline{d}_s^1 = \min \{d_{\rm I}+\hat d_{\rm II}+\xi^\star,\max \{N_s^2,N_d^1\} - N_s^2\}.
\end{align}
We should note that this expression also applies to the case of $N_s^2 \ge N_d^1$, where $\underline{d}_s^1 =0$.

Summarizing the above two cases, we arrive at
\begin{align}
\underline{d}_s^1 = \left\{ {\begin{array}{*{20}{c}}
{\min \{ d_{\rm I} + d_{\rm III} + \beta^ \star, \eta - N_d^2,N_d^1\} }, {\rm {if}}\ {N_s^2 > N_d^2}\\
\min \{d_{\rm I}+\hat d_{\rm II}+\xi^\star,\eta - N_s^2\},{\rm {if}}\ {N_s^2 \le N_d^2}
\end{array}} \right. \label{eq522}
\end{align}
where $\eta=\max \{N_s^2,N_d^1\}$.

\section{Construction of Precoding Matrices Which Achieve the Point on the S.D.o.F. Region Boundary}
\begin{table}[!htp]
\caption{An algorithm for constructing $({\bf V}, {\bf W})$
which achieve $(\hat d_s^1, \hat d_s^2)$ on the  S.D.o.F. region boundary}
 \begin{tabular}{p{0.95\linewidth}}
  \hline
1. Initialize $u=\min\{\hat d_s^1, \min \{y_{\max}, d_{\rm V}\}+d_{\rm I}+d_{\rm III}\}$,
$t=\hat d_s^1-u$;
 \\\\
2. $({\bf V}_{\rm{o}},{\bf W}_{\rm{o}})\leftarrow$ select $u$ precoding vector pairs from $Sub_{\rm{o}}$; \\
3. $({\bf V}_{\rm{e}},{\bf W}_{\rm{e}})\leftarrow$ select $t$ precoding vector pairs from $Sub_{\rm{e}}$; \\
4. ${\bf V}\leftarrow [{\bf V}_{\rm{o}} \ {\bf V}_{\rm{e}}]$; \\
5. ${\bf W}_1\leftarrow [{\bf W}_{\rm{o}} \ {\bf W}_{\rm{e}}]$; \\ \\
6. Let $\breve d_s^2=\hat d_s^2- {\rm {rank}}({\bf W}_1)$;\\
7. \textbf{if $\breve d_s^2>0$} \\
8. \quad Let $\tilde d_s^2= \min \{\breve d_s^2, (N_s^2-N_d^1)^+\}$; \\
9. \quad ${\bf W}_2 \leftarrow {\bf A}(:, 1:\tilde d_s^2)$, where ${\bf A}={\bf \Gamma}({\bf H}_{12})$;  \\
10. \quad Do the singular value decomposition (SVD) ${\bf H}_{22}={\bf USR}^H$;\\
11. \quad ${\bf W} \leftarrow [{\bf W}_1\ {\bf W}_2\ {\bf R}(:,1:\breve d_s^2-\tilde d_s^2)]$; \\
12. \textbf{else}    \\
13. \quad ${\bf W}\leftarrow {\bf W}_1$;  \\
14. \textbf{end}
\\\\
15. \textbf{Output:} $({\bf V}, {\bf W})$.
\\

\hline
\end{tabular}
\end{table}

According to Section V. C, by carefully choosing
$({\bf v}, {\bf w})$ we are able
to construct precoding matrix pairs $({\bf V}, {\bf W})$ which achieve the S.D.o.F. pairs on the  S.D.o.F. region boundary.
In particular, by selecting $u=\min\{\hat d_s^1, \min \{y_{\max}, d_{\rm V}\}+d_{\rm I}+d_{\rm III}\}$ pairs from
$Sub_{\rm{o}}=Sub_{\rm I}\cup Sub_{\rm III} \cup Sub_{\rm V}$
and $t=\hat d_s^1-u$ pairs from $Sub_{\rm{e}}=Sub_{\rm II}\cup Sub_{\rm IV} \cup Sub_{\rm VI}$,
subject to the number of pairs selected from $Sub_{\rm V} \cup Sub_{\rm VI}$ being no greater than $y_{\max}$, we have completed
the construction of precoding matrices $({\bf V}, {\bf W}(:, 1:K_v)) \in \hat{\mathcal I}$.
This construction satisfies $d_s^1=\hat d_s^1$ and also leaves a maximum dimension, i.e., $d_s^2=\hat d_s^2$ (see eq. (\ref{eq74})),
interference-free subspace for \emph{D}$_2$.
Further, if $\hat d_s^2 \le {\rm {rank}}({\bf W}(:, 1:K_v))$,
\emph{S}$_2$ does not need to add any beamforming vectors, and the S.D.o.F. of $\hat d_s^2$ is achieved.
In this case, $K_w$ equals the number of nonzero columns of ${\bf W}(:, 1:K_v)$. If $\hat d_s^2 > {\rm {rank}}({\bf W}(:, 1:K_v))$,
\emph{S}$_2$ can add $\breve d_s^2=\hat d_s^2- {\rm {rank}}({\bf W}(:, 1:K_v))$ columns to its precoding matrix without
violating any constraints of $\hat{\cal I}$ and also achieves an S.D.o.F. of $\hat d_s^2$.
In particular, by adding the first $\tilde d_s^2 = \min \{\breve d_s^2, (N_s^2-N_d^1)^+\}$ columns of ${\bf \Gamma}({\bf H}_{12})$ and
the first $\breve d_s^2- \tilde d_s^2$ columns of $\bf R$ as the other beamforming vectors
at \emph{S}$_2$, we complete the construction of the precoding matrices $({\bf V}, {\bf W})$. In this case $K_w=\hat d_s^2$.
Here $\bf R$ is obtained with the singular value decomposition (SVD) ${\bf H}_{22}={\bf USR}^H$.
By this SVD the channel ${\bf H}_{22}$ is decomposed into several parallel sub-channels, and
the first $\breve d_s^2- \tilde d_s^2$ columns of $\bf R$ correspond the ones
which are of better channel quality than the others.

\emph{Example 4:}
Let us revisit \emph{Example 3}, in which we obtained
an S.D.o.F. pair $(\hat d_s^1,\hat d_s^2)=(2,4)$ on the strict S.D.o.F. region boundary.
According to Section V. C, at this boundary point, $y_{\max}=2$ and $z_{\min}=0$.
Since $u=2$, $d_{\rm I}=d_{\rm III}=0$ and $d_{\rm V}=2$, we first select two precoding
vector pairs in $Sub_{\rm V}$, i.e., $({\bf v}_1,{\bf w}_1)$ and $({\bf v}_2,{\bf w}_2)$,
with ${\bf H}_{21}{\bf v}_1=0$, ${\bf H}_{21}{\bf v}_2=0$, ${\bf H}_{12}{\bf w}_1 \ne 0$ and ${\bf H}_{12}{\bf w}_2 \ne 0$.
From the remaining sets we do not pick any pairs since $t=0$.
So far, we have finished the construction of ${\bf V}$ and $ {\bf W}(:, 1:K_v)$, i.e., $[{\bf v}_1\ {\bf v}_2]$ and $[{\bf w}_1\ {\bf w}_2]$.
Since $\breve d_s^2=\hat d_s^2- {\rm {rank}}({\bf W}(:, 1:K_v))=2 >0$,
we further add $\tilde d_s^2= \min \{\breve d_s^2, (N_s^2-N_d^1)^+\}=1$ column of ${\bf \Gamma}({\bf H}_{12})$,
i.e., ${\bf w}_3$, with ${\bf H}_{12}{\bf w}_3=0$,
and $\breve d_s^2-\tilde d_s^2=1$ column of $\bf R$, i.e., ${\bf w}_4$, with ${\bf H}_{22}{\bf w}_4 \ne 0$,
as the other beamforming vectors at \emph{S}$_2$.
Since ${\bf H}_{11}{\bf v}_i \ne 0$, ${\bf H}_{22}{\bf w}_i \ne 0$ and ${\bf H}_{12}{\bf w}_4 \ne 0$
hold true with probability one,
for ${\bf V}=[{\bf v}_1\ {\bf v}_2]$ and ${\bf W}=[{\bf w}_1\ {\bf w}_2\ {\bf w}_3\ {\bf w}_4]$
it holds that ${\rm dim}\{{\rm span}({\bf H}_{11}{\bf V})\}+{\rm dim}\{{\rm span}({\bf H}_{12}{\bf W})\}=2+3=5$
and ${\rm dim}\{{\rm span}({\bf H}_{22}{\bf W})\}+{\rm dim}\{{\rm span}({\bf H}_{21}{\bf V})\}=4+0=4$.
Therefore, the S.D.o.F. pair $(\hat d_s^1,\hat d_s^2)=(2,4)$ is achieved.

Concluding, an algorithm for constructing $({\bf V}, {\bf W})$ is given in TABLE III.
Note that the formulas of ${\bf v}_i$ and ${\bf w}_i$ in $Sub_i$, $i={\rm I},{\rm II},\cdots,{\rm VI}$, are given in
(\ref{eq52a1}), (\ref{eq54a2}), (\ref{eq57a3}), (\ref{eq61a2}), (\ref{eq64a2}) and (\ref{eq67a2}), respectively.

\emph{Remark 3}: In light of (\ref{eq18a}) and (\ref{eq18b}) derived in \emph{Proposition 2},
whenever we find a solution $({\bf V}, {\bf W})$
achieving the S.D.o.F. pair $(\hat d_s^1, \hat d_s^2)$ on the  S.D.o.F. region boundary, we
actually find the solution spaces ${\rm {span}}({\bf V})$ and $ {\rm {span}}({\bf W})$, i.e.,
the precoding matrices $({\bf VA}, {\bf WB})$ also achieve
the S.D.o.F. pair $(\hat d_s^1, \hat d_s^2)$ on the  S.D.o.F. region boundary as long as
$\bf A$ and $\bf B$ are invertible.

\section{Numerical Results}

In this section, we give numerical results to validate our theoretical findings.
For simplicity, we consider a simple semi-symmetric system model, as illustrated in Fig. 3.
In particular, the antenna numbers $N_s^1=N_d^1=N_e \triangleq N_1 $, and $N_s^2=N_d^2 \triangleq N_2 $.
We assume that \emph{D}$_i$ or \emph{E} is uniformly distributed on a ring of radius $1 \le R \le 10$ (unit: meters) and center located at \emph{S}$_i$.
The source-destination distances or the source-eavesdropper distance
are no greater than the source-source distance. To highlight the effects of
distances, the channel between any transmit-receiver antenna pair is modeled by a
simple line-of-sight channel model including the path loss effect
and a random phase, i.e., $h_{12}=d_{12}^{-c/2}e^{j\theta}$ where $d_{12}$ denotes the distance
between the \emph{S}$_2$ and \emph{D}$_1$, $c=3.5$ is the path loss exponent, $\theta$ is the random phase uniformly distributed
within $[0, 2 \pi)$.
The distances between transmit or receiver antennas at each terminal are assumed
to be much smaller than the source-destination distance or the source-eavesdropper distance,
so the path losses of different transmit-receiver antenna pairs from the same transmit-receiver link
are approximately the same. \emph{S}$_2$ is located at a fixed two-dimensional coordinates
(0,0) (unit: meters), while \emph{S}$_1$ moves from (350,0) to (10,0).
The transmitting power of each source is $P= 0$dBm.
Results are averaged over one hundred thousand independent channel trials.

\begin{figure}[!t]
\centering
\includegraphics[width=3in]{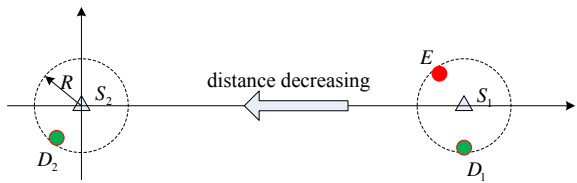}
\DeclareGraphicsExtensions. \caption{Model used for numerical experiments}
\vspace* {-6pt}
\end{figure}

\begin{figure}[!t]
\centering
\includegraphics[width=3in]{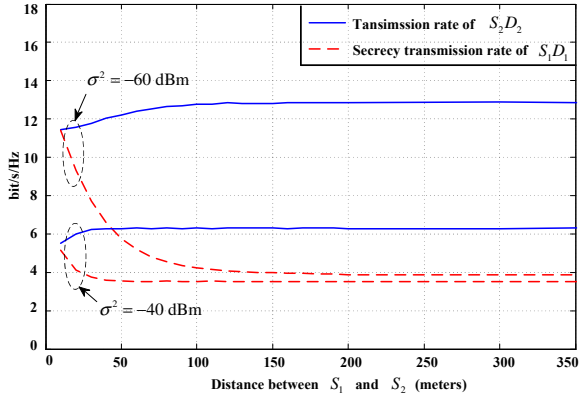}
\DeclareGraphicsExtensions. \caption{Achievable rates versus $S_1$-$S_2$ distance}
\vspace* {-6pt}
\end{figure}

Fig. 4 illustrates the achievable secrecy transmission rate of the user \emph{S}$_1$-\emph{D}$_1$,
and also the achievable transmission rate of the user \emph{S}$_2$-\emph{D}$_2$ for $N_1=4$ and $N_2=2$.
The noise power $\sigma^2=-60$dBm and $\sigma^2=-40$dBm are considered, respectively. According
to (\ref{eq74}), we see that with our proposed cooperative transmission
scheme, the S.D.o.F. pair (1,1) can be achieved. We compute the precoding vectors $\bf v$ and $\bf w$
according to TABLE III, and compute the achievable
transmission rate of each user according to (\ref{eq4}) and (\ref{eq5}), respectively.
It shows that the achievable secrecy transmission rate of \emph{S}$_1$-\emph{D}$_1$ increases monotonically as
\emph{S}$_1$ moves close to \emph{S}$_2$. In contrast, the achievable transmission rate of \emph{S}$_2$-\emph{D}$_2$
decreases with the decreasing of the source-source distance. As compared with the decrease in the transmission rate of
\emph{S}$_2$-\emph{D}$_2$, the increase in the secrecy transmission rate of \emph{S}$_1$-\emph{D}$_1$ is drastic.
Therefore, the network performance benefits when the two users get closer.

\begin{figure}[!t]
\centering
\includegraphics[width=3in]{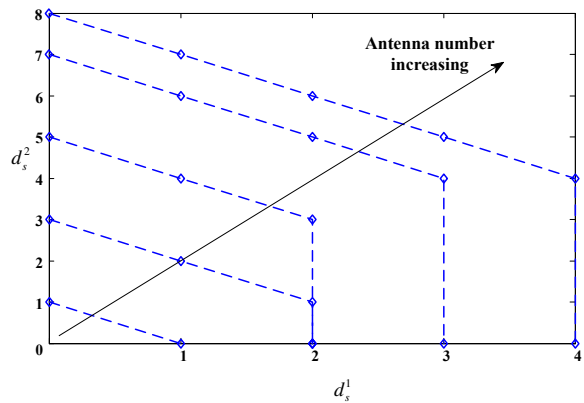}
\DeclareGraphicsExtensions. \caption{Achievable secrecy degrees of freedom region with an increasing number of antennas at \emph{S}$_2$-\emph{D}$_2$}
\vspace* {-6pt}
\end{figure}

Fig. 5 illustrates the achievable secrecy degrees of freedom region versus
different values of $N_2$. Here, we set $N_1=4$ and let $N_2$ vary from 1 to 8.
We compute the achievable secrecy degrees of freedom region according to (\ref{eq74}).
As expected, the secrecy degrees of freedom region expands with an increasing $N_2$.
Note that previous work \cite{Wornell11} shows that for the
classic wiretap channel with no cooperative helpers the
condition to achieve a nonzero S.D.o.F. is $N_s^1 \ge N_e+1$.
Here although $N_s^1 = N_e$,
by exploiting the co-channel interference an S.D.o.F. of $N_s^1$ can be achieved.

\begin{figure}[!t]
\centering
\includegraphics[width=3in]{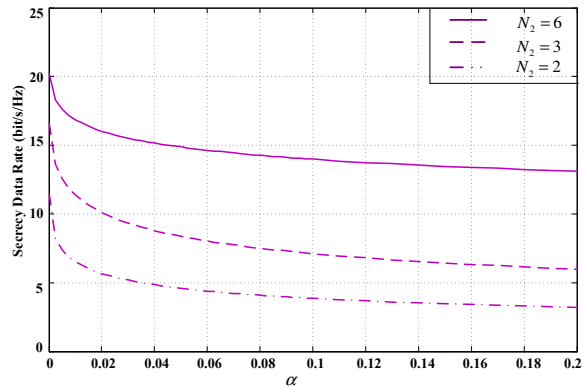}
\DeclareGraphicsExtensions. \caption{Achievable secrecy rate of \emph{S}$_1$-\emph{D}$_1$ versus the uncertainty of the eavesdropper's channels $\alpha$. }
\vspace* {-6pt}
\end{figure}

In practice, while one may have a good estimate of the position of the eavesdropper, an estimate of
the phase of the eavesdropper's channels is more difficult to obtain.
Since the proposed precoding matrix design highly depends on the eavesdropper's
channels, we next examine the secrecy rate performance degradation in the presence of imperfect channel estimate.
In Fig. 6, we plot the achievable secrecy rate with imperfect CSI of the eavesdropper's channels.
Here, we set $N_1=4$ and let $N_2$ vary from 2 to 6. $S_1$ and $S_2$ are located at
(10,0) and (0,0), respectively. The noise power $\sigma^2=-60$dBm.
The channel from \emph{S}$_i$ ($i=1,2$) to \emph{E} is
\begin{align}
{\bf G}_{i}=d_{ei}^{-c/2}\left(\dfrac{1}{\sqrt{1+\alpha}}\bar{\bf G}_{i}+\sqrt{\dfrac{\alpha}{1+\alpha}}\Delta\bar{\bf G}_{i}\right),
\end{align}
where $\alpha$ denotes the channel uncertainty.
$\bar{\bf G}_{i}$ represents the estimated channel part at \emph{S}$_i$.
The entries of $\bar{\bf G}_{i}$ are $e^{j\theta}$ with $\theta$ be a random phase
uniformly distributed within $[0, 2 \pi)$.
$\Delta\bar{\bf G}_{i}\sim \mathcal{CN}(\bf{0},\bf{I})$
represents the Gaussian error channel matrices.
$d_{ei}$ denotes the distance from \emph{S}$_i$.
According to (\ref{eq74}), we see that the S.D.o.F. pairs (1,1), (2,1) and (3,3) can be achieved for the case of $N_2=2$, $N_2=3$ and
$N_2=6$, respectively. For these S.D.o.F. pairs, we construct the precoding matrices $\bf V$ and $\bf W$
according to TABLE III, subject to power being equally allocated between different signal streams.
The achievable secrecy transmission rate is computed according to (\ref{eq4}).
It can be observed that the achievable secrecy rate drops with the
increase of channel uncertainties when
the channel uncertainty $\alpha$ is small. Fortunately, when the number of antennas
$N_2$ increases, this secrecy rate performance degradation is smaller.
On the other hand, on comparing the secrecy transmission rate of \emph{S}$_1$-\emph{D}$_1$ for the case $N_2=2$ with that in Fig. 4,
one can see that the secrecy rate achieved for the case where $\alpha=0.1$ and
\emph{S}$_1$-\emph{S}$_2$ distance of 10 meters, is almost equal to the
secrecy rate achieved for the case where $\alpha=0$ and \emph{S}$_1$-\emph{S}$_2$ distance of 150 meters.
This suggests that in wiretap interference networks, the secrecy rate
degradation due to CSI estimation error can be counteracted by bringing the two users closer together.

\section{Conclusion}
We have examined the maximum achievable secrecy degrees
of freedoms (S.D.o.F.) region of a MIMO two-user wiretap interference channel,
where one user requires confidential connection against an external passive eavesdropper, while the other uses a public connection.
We have addressed analytically the S.D.o.F. pair maximization (component-wise).
Specifically, we have proposed a cooperative secrecy transmission scheme and proven
that its feasible set is sufficient to achieve all the points on the S.D.o.F. region boundary.
For the proposed cooperative secrecy transmission scheme,
we have obtained analytically the maximum achievable
S.D.o.F. region boundary points. We have also constructed the precoding matrices which achieve the S.D.o.F. region boundary.
Our results revealed the connection between
the maximum achievable S.D.o.F. region and the number of antennas, thus
shedding light on how the secrecy rate region 
behaves for different number of antennas.
Numerical results show that the network performance benefits when the two users get closer.
This is interesting. It tells us that in wiretap interference networks, the secrecy rate
degradation due to CSI estimation error can be counteracted by bringing the two users closer together.

\appendices
\section{Proof of \emph{Proposition 1}} \label{appD}
In what follows, we prove that ${\bf{A}}{\bf{v}}={\bf{B}}{\bf{w}}$ holds true
if and only if $\bf v$ and $\bf w$ are given in (\ref{eq80a}) and (\ref{eq80b}),
with ${\bf{y}}_s$, ${\bf{y}}_{s1}$, ${\bf{y}}_{s2}$,
${\bf{y}}_1$ and ${\bf{y}}_2$ being any vectors with appropriate length.
With this result, the first conclusion in \emph{Proposition 1} is a natural extension.
According to the GSVD decomposition,
${\bf{A}}{\bf \Psi}_{12}{\bf \Lambda}_1^{-1}={\bf{B}}{\bf \Psi}_{22}{\bf \Lambda}_2^{-1}={\bf X}_2$.
Thus, ${\bf{A}}{\bf{v}}={\bf{B}}{\bf{w}} $ holds true if ${\bf v}$
and ${\bf w}$ are given by (\ref{eq80a}) and (\ref{eq80b}), respectively.
Next, we prove by contradiction that
${\bf{A}}{\bf{v}}={\bf{B}}{\bf{w}} $ holds true only if
${\bf v} \in {\rm span}({\bf \Phi}_1)$; the argument for $\bf w$ is similar. Assume that there exists a nonzero vector
$\bar {\bf v} \notin {\rm span}({\bf \Phi}_1)$ satisfying ${\bf{A}}\bar{\bf{v}}={\bf{B}}{\bf{w}}$.
Then, ${\bf A}\bar {\bf v} \notin {\rm span}({\bf A}{\bf \Phi}_1)$; otherwise, it holds that
${\bf A}\bar {\bf v} ={\bf A}{\bf \Phi}_1{\bf x}$ which implies $\bar {\bf v} -{\bf \Phi}_1{\bf x}={\bf \Gamma}({\bf A}){\bf y}_1 $,
and so $\bar {\bf v} \in {\rm span}({\bf \Phi}_1)$ which contradicts with the assumption.
However, ${\bf A}\bar {\bf v} \in {\rm {span}}({\bf X}_2)$
due to ${\bf{A}}\bar{\bf{v}}={\bf{B}}{\bf{w}}$. In addition, by the GSVD,
${\rm {span}}({\bf X}_2)={\rm span}({\bf A}{\bf \Phi}_1)$. Thus, ${\bf A}\bar {\bf v} \in  {\rm span}({\bf A}{\bf \Phi}_1)$
and so ${\bf A}\bar {\bf v} \notin {\rm span}({\bf A}{\bf \Phi}_1)$ is contradicted.
This completes the proof of the first conclusion in \emph{Proposition 1}.

According to the GSVD, ${\bf{A}}{\bf{\Psi}}_{13}={\bf 0}$. Thus, ${\rm span}({\bf{\Psi}}_{13}) \subset {\rm span}({\bf \Gamma}({\bf A}))$.
In addition, ${\rm rank}({\bf{\Psi}}_{13})=M-r-s=M-\min \{M,N\}=(M-N)^+$, which indicates that the linear independent vectors in ${\rm span}({\bf{\Psi}}_{13})$ is the same as that in ${\rm span}({\bf \Gamma}({\bf A}))$. So, ${\rm span}({\bf{\Psi}}_{13}) = {\rm span}({\bf \Gamma}({\bf A}))$. Since ${\bf{\Psi}}_1$ is an unitary matrix, it holds that ${\rm span}({\bf{\Psi}}_{12}) \cap {\rm span}({\bf{\Psi}}_{13}) ={\bf 0}$.
Therefore, ${\rm span}({\bf{\Psi}}_{12}) \cap {\rm span}({\bf \Gamma}({\bf A})) ={\bf 0}$, which, combined with (\ref{eq80a}), indicates that
the number of linearly independent vectors $\bf v$ satisfying ${\bf{A}}{\bf{v}}={\bf{B}}{\bf{w}} \ne {\bf 0} $
is $s+{\rm {dim}} \{{\rm {null}}({\bf{A}}) \}$. This completes the proof.

\section{Proof of \emph{Proposition 2}} \label{appA}
Given an arbitrary point $({\bf V},{\bf W})$, with
${\rm{tr}} \{{\bf Q}_v\}=P$ and $ {\rm{tr}} \{{\bf Q}_w\}=P$.
We can respectively rewrite ${\bf Q}_v$ and ${\bf Q}_w$ as
${\bf Q}_v=  P \bar{\bf Q}_v$ and ${\bf Q}_w=  P \bar{\bf Q}_w $,
with ${\rm{tr}} \{\bar{\bf Q}_v\}={\rm{tr}}\{{\bar{\bf Q}_w}\} = 1$.
Correspondingly, (\ref{eq3a}) can be rewritten as
\begin{align}
R_d^1={\rm {log}}|{\bf I}+({\bf I}+  P{\bf H}_{12}\bar {\bf Q}_w{\bf H}_{12}^H)^{-1}  {\bf H}_{11}\bar{\bf Q}_v{\bf H}_{11}^HP|\label{eq20}.
\end{align}

Let ${\bf\Theta}_2={\bf H}_{11}\bar{\bf Q}_v{\bf H}_{11}^H$.
Denoting ${\bf H}_{12}\bar {\bf Q}_w{\bf H}_{12}^H=\left[{\bf U}_{1}\ {\bf U}_{0} \right]
\left[ {\begin{array}{*{20}{c}}
{{{\bf{\Sigma}}_{1}}}&{\bf{0}}\\
{\bf{0}}&{{{\bf{0}}}}
\end{array}} \right]
\left[ {\begin{array}{*{20}{c}}
{{{\bf{U}}_{1}^H}}\\
{{{\bf{U}}_{0}^H}}
\end{array}} \right]$ as the singular value decomposition (SVD),
and substituting it into (\ref{eq20}), we obtain
\begin{align}
R_d^1 = &{\rm {log}}|{\bf I}+{\bf U}_{1}({\bf I}+  P{\bf \Sigma}_1)^{-1}{\bf U}_{1}^H
  {\bf\Theta}_2P+{\bf U}_{0}{\bf U}_{0}^H  {\bf\Theta}_2P|\nonumber \\
=& {\rm {log}}|{\bf I}+{\bf U}_{1}(\frac {\bf I}{P}+  {\bf \Sigma}_1)^{-1}{\bf U}_{1}^H
  {\bf\Theta}_2+{\bf U}_{0}{\bf U}_{0}^H {\bf\Theta}_2 P|. \nonumber
\end{align}
Therefore,
\begin{align}
&\lim_{P \to \infty} {R_d^1}/{{\rm{log}}(P)} \nonumber \\
=&\lim_{P \to \infty}\dfrac{{\rm {log}}|{\bf I}+{\bf U}_{1}(  {\bf \Sigma}_1)^{-1}{\bf U}_{1}^H
  {\bf\Theta}_2+{\bf U}_{0}{\bf U}_{0}^H {\bf\Theta}_2 P|}{{\rm{log}}(P)} \nonumber \\
=&\lim_{P \to \infty}\dfrac{{\rm {log}}|{\bf I}+(\frac{1}{P}{\bf U}_{1}(  {\bf \Sigma}_1)^{-1}{\bf U}_{1}^H
  +{\bf U}_{0}{\bf U}_{0}^H) {\bf\Theta}_2 P|}{{\rm{log}}(P)} \nonumber \\
=&\lim_{P \to \infty}\dfrac{{\rm {log}}|{\bf I}+{\bf U}_{0}{\bf U}_{0}^H {\bf H}_{11}{\bf V}{\bf V}^H{\bf H}_{11}^H|}{{\rm{log}}(P)} \nonumber \\
= & {\rm {rank}}\{{\bf U}_{0}{\bf U}_{0}^H{\bf H}_{11}{\bf V}{\bf V}^H{\bf H}_{11}^H\}  \label{eq21b} \\
= & {\rm{dim}}\{{\rm{span}}({\bf{H}}_{11}{\bf{V}})/{\rm{span}}({\bf{H}}_{12}{\bf{W}})\}\label{eq21}\\
=&{\rm {rank}}\{{\bf{H}}_{11}{\bf V}\}-
{\rm{dim}}\{{\rm{span}}({\bf{H}}_{11}{\bf{V}})\cap{\rm{span}}({\bf{H}}_{12}{\bf{W}})\}.\label{eq23}
\end{align}
where (\ref{eq21b}) comes from the fact that
\begin{align}
\lim_{P \to \infty}\dfrac{{\rm {log}}|{\bf I}+{\bf A}P|}{{\rm{log}}(P)}
=\lim_{P \to \infty}\dfrac {\sum\nolimits_{i = 1}^t {{\rm{log}}(1+{\lambda _i}P)}}{{\rm{log}}(P)}
= {\rm {rank}}\{\bf A\}, \nonumber
\end{align}
with ${\lambda _i}$ and $t$ being the nonzero eigenvalue and the rank of ${\bf A}$.
(\ref{eq21}) comes from the fact that ${\bf U}_{0}{\bf U}_{0}^H$ is the projection matrix of
the subspace ${\rm {span}}({\bf{H}}_{12}{\bf{W}})^\perp$.

Applying similar derivations from (\ref{eq20}) to (\ref{eq21}) yields
\begin{align}
\lim_{P \to \infty} \dfrac{R_d^2}{{\rm{log}}(P)}
={\rm{dim}}\{{\rm{span}}({\bf{H}}_{22}{\bf{W}})/{\rm{span}}({\bf{H}}_{21}{\bf{V}})\}, \label{eq24}\\
\lim_{P \to \infty} \dfrac{R_e}{{\rm{log}}(P)}
={\rm{dim}}\{{\rm{span}}({\bf{G}}_{1}{\bf{V}})/{\rm{span}}({\bf{G}}_{2}{\bf{W}})\}. \label{eq25}
\end{align}
Substituting (\ref{eq23})-(\ref{eq25}) into (\ref{eq8}), we arrive at (\ref{eq18a}) and (\ref{eq18b}).
This completes the proof.

\section{Proof of \emph{Proposition 3}} \label{appB}
By definition, we have $\bar{\mathcal{D}} \subset {\mathcal{D}}$. Thus, the boundary of $\bar{\mathcal{D}}$
is included by that of ${\mathcal{D}}$.
In the following, we show that
for any given precoding matrices $({\bf V}, {\bf W}) \in {\mathcal I}$,
we can always find another precoding matrices $({\bf V}^\prime,{\bf W}^\prime) \in \bar{\mathcal I}$,
which satisfy $d_s^1({\bf V},{\bf W}) \le d_s^1({\bf V}^\prime,{\bf W}^\prime)$ and
$d_s^2({\bf V},{\bf W})\le d_s^2({\bf V}^\prime,{\bf W}^\prime)$. So, the boundary of ${\mathcal{D}}$
is included by that of $\bar{\mathcal{D}}$. Concluding,
the outer boundary of ${\mathcal{D}}$ is the same as that of $\bar{\mathcal{D}}$.

Before proceeding, we first introduce two critical properties on matrix
that will be used in the following analyses. That is, for any given matrices $\bf A$ and $\bf B$,
if $\bf B$ is invertible, then
\begin{align}
&{\rm {span}}({\bf A})={\rm {span}}({\bf A}{\bf B}),  \label{eq26}  \\
& {\rm {rank}}\{{\bf A}\}={\rm {rank}}\{{\bf A}{\bf B}\}\label{eq27}.
\end{align}

In what follows, based on the GSVD decomposition of $({\bf H}_{12}{\bf W},{\bf H}_{11}{\bf V})$
we first construct a precoding matrix pair $(\hat{\bf V},\hat{\bf W})$, which excludes the
intersection subspace of ${\rm{span}}({\bf H}_{12}{\bf W})$ and ${\rm{span}}({\bf H}_{11}{\bf V})$
without decreasing the achieved S.D.o.F. pair. Further, based on the
GSVD decomposition of $({\bf G}_2\hat{\bf W},{\bf G}_1\hat{\bf V})$ we construct a
precoding matrix pair $({\bf V}^\prime,{\bf W}^\prime)$, which excludes the subspace
${\rm{span}}({\bf G}_{21}\hat{\bf V})/{\rm{span}}({\bf G}_{22}\hat{\bf W})$ without decreasing the achieved S.D.o.F. pair.
In this way, we finish the construction of the wanted precoding matrix pair.


Consider the decomposition
\begin{align}
{\rm {GSVD}}&({\bf H}_{12}{\bf W},{\bf H}_{11}{\bf V};N_d^1,K_w,K_v) \nonumber\\
&=(\hat {\bf\Psi}_1,\hat{\bf\Psi}_2,\hat{\bf \Lambda}_1,\hat{\bf \Lambda}_2,\hat{\bf X},\hat k, \hat r,\hat s,\hat p). \label{eq28}
\end{align}
Let $\hat{\bf\Psi}_2^0=[\hat{\bf\Psi}_{21},\hat{\bf\Psi}_{23}]$.
Since $\hat{\bf\Psi}_1$ and $\hat{\bf\Psi}_2$ are invertible,
$\hat{\bf\Psi}_1^\prime=[\hat{\bf\Psi}_{11},\hat{\bf\Psi}_{13},\hat{\bf\Psi}_{12}]$ and
$\hat{\bf\Psi}_2^\prime=[\hat{\bf\Psi}_2^0,\hat{\bf\Psi}_{22}]$ are also invertible.
Applying (\ref{eq26}) and (\ref{eq27}), we have
\begin{subequations}
\begin{align}
d_s^1({\bf V},{\bf W})&=d_s^1({{\bf V}\hat{\bf\Psi}_2^\prime},{{\bf W}\hat{\bf\Psi}_1^\prime}) \label{eq30a}\\
& = {\rm {rank}}\{{\bf{H}}_{11}{\bf V}\hat{\bf\Psi}_2^0\}
-m({\bf V}\hat{\bf\Psi}_2^\prime,{\bf W}\hat{\bf\Psi}_1^\prime) \label{eq30b}\\
&\le {\rm {rank}}\{{\bf{H}}_{11}{\bf V}\hat{\bf\Psi}_2^0\}
-m({\bf V}\hat{\bf\Psi}_2^0,{\bf W}\hat{\bf\Psi}_1^\prime), \label{eq30c}
\end{align}
\end{subequations}
in which (\ref{eq30b}) can be justified with ${\rm{span( }}{{\bf{H}}_{12}}{\bf{W}}\hat{\bf\Psi}_1^\prime
{\rm{) }} \cap {\rm{span}}({{\bf{H}}_{11}}{\bf{V}}\hat{\bf\Psi}_2^\prime)
={\rm{span}}({{\bf{H}}_{11}}{\bf{V}}\hat{\bf\Psi}_{22}$).
Besides, (\ref{eq30c}) comes from the fact that
$m({\bf V}\hat{\bf\Psi}_2^\prime,{\bf W}\hat{\bf\Psi}_1^\prime)
\ge m({\bf V}\hat{\bf\Psi}_2^0,{\bf W}\hat{\bf\Psi}_1^\prime)$.
Here $({\bf V}\hat{\bf\Psi}_2^0,{\bf W}\hat{\bf\Psi}_1^\prime)$
is the precoding matrix pair $(\hat {\bf V}, \hat {\bf W})$ we mentioned in the above text.

Consider the decomposition
\begin{align}
{\rm {GSVD}}&({\bf G}_{2}{\bf W}\hat{\bf\Psi}_1^\prime,{\bf G}_{1}{\bf V}\hat{\bf\Psi}_2^0;N_e,K_w,K_v-\hat s) \nonumber \\
&=(\breve {\bf\Psi}_1,\breve{\bf\Psi}_2,\breve{\bf \Lambda}_1,\breve{\bf \Lambda}_2,\breve{\bf X},\breve k, \breve r,\breve s,\breve p). \label{eq31}
\end{align}
Let $\breve{\bf\Psi}_2^1=[\breve{\bf\Psi}_{21}, \breve{\bf\Psi}_{22}]$.
Since $\breve{\bf\Psi}_1$ and $\breve{\bf\Psi}_2$ are invertible,
$\breve{\bf\Psi}_1^\prime=[\breve{\bf\Psi}_{13}, \breve{\bf\Psi}_{11}, \breve{\bf\Psi}_{12}]$ and
$\breve{\bf\Psi}_2^\prime=[\breve{\bf\Psi}_{23}, \breve{\bf\Psi}_2^1]$ are also invertible.
Applying (\ref{eq26}) and (\ref{eq27}), we have
\begin{subequations}
\begin{align}
& {\rm {rank}}\{{\bf{H}}_{11}{\bf V}\hat{\bf\Psi}_2^0\}
-m({\bf V}\hat{\bf\Psi}_2^0,{\bf W}\hat{\bf\Psi}_1^\prime) \nonumber \\
&={\rm {rank}}\{{\bf{H}}_{11}{\bf V}\hat{\bf\Psi}_2^0\breve{\bf\Psi}_2^\prime\}
-m({\bf V}\hat{\bf\Psi}_2^0\breve{\bf\Psi}_2^\prime,{\bf W}\hat{\bf\Psi}_1^\prime\breve{\bf\Psi}_1^\prime) \label{eq33a}\\
& ={\rm {rank}}\{{\bf{H}}_{11}{\bf V}\hat{\bf\Psi}_2^0\breve{\bf\Psi}_2^\prime\}
-{\rm {rank}}\{\breve{\bf\Psi}_{23}\} \label{eq33b}\\
& \le {\rm {rank}}\{{\bf{H}}_{11}{\bf V}\hat{\bf\Psi}_2^0\breve{\bf\Psi}_2^1\}.  \label{eq33c}
\end{align}
\end{subequations}
Here, since ${\rm{span}}({\bf{G}}_{1}{\bf{V}}\hat{\bf\Psi}_2^0\breve{\bf\Psi}_2^\prime)
/{\rm{span}}({\bf{G}}_{2}{\bf{W}}\hat{\bf\Psi}_1^\prime\breve{\bf\Psi}_1^\prime)
={\rm{span}}({\bf{G}}_{1}{\bf{V}}\hat{\bf\Psi}_2^0\breve{\bf\Psi}_{23})={\rm {rank}}\{\breve{\bf\Psi}_{23}\}$,
we see that (\ref{eq33b}) holds true.
Since ${\rm {rank}}\{{\bf{H}}_{11}{\bf V}\hat{\bf\Psi}_2^0\breve{\bf\Psi}_{23}\}
\le {\rm {rank}}\{\breve{\bf\Psi}_{23}\}$ and ${\rm {rank}}\{{\bf{H}}_{11}{\bf V}\hat{\bf\Psi}_2^0\breve{\bf\Psi}_2^\prime\}
\le {\rm {rank}}\{{\bf{H}}_{11}{\bf V}\hat{\bf\Psi}_2^0\breve{\bf\Psi}_2^1\}
+{\rm {rank}}\{{\bf{H}}_{11}{\bf V}\hat{\bf\Psi}_2^0\breve{\bf\Psi}_{23}\}$, we see that (\ref{eq33c}) holds true.

Combining (\ref{eq30a})-(\ref{eq30c}) with (\ref{eq33a})-(\ref{eq33c}), we arrive at
\begin{align}
d_s^1({\bf V},{\bf W}) \le {\rm {rank}}\{{\bf{H}}_{11}{\bf V}\hat{\bf\Psi}_2^0\breve{\bf\Psi}_2^1\}\label{eq34}.
\end{align}

On the other hand,
according to (\ref{eq31}), it holds that
$m({\bf V}\hat{\bf\Psi}_2^0\breve{\bf\Psi}_2^1,{\bf W}\hat{\bf\Psi}_1^\prime\breve{\bf\Psi}_1^\prime)=0$, which indicates
\begin{align}
{\rm{span}}({\bf{G}}_{1}{\bf V}\hat{\bf\Psi}_2^0\breve{\bf\Psi}_2^1)
\subset {\rm{span}}({\bf{G}}_{2}{\bf W}\hat{\bf\Psi}_1^\prime\breve{\bf\Psi}_1^\prime)\label{eq35}.
\end{align}
According to (\ref{eq28}),
${\rm{span}}( {\bf{H}}_{12}{\bf W}\hat{\bf\Psi}_1^\prime) \cap {\rm{span}}({\bf H}_{11}{\bf V}\hat{\bf\Psi}_2^0)={\bf 0}$,
which together with
${\rm{span}}( {\bf{H}}_{12}{\bf W}\hat{\bf\Psi}_1^\prime)={\rm{span}}( {\bf{H}}_{12}{\bf W}\hat{\bf\Psi}_1^\prime\breve{\bf\Psi}_1^\prime)$ and
${\rm{span}}({\bf H}_{11}{\bf V}\hat{\bf\Psi}_2^0) \supset {\rm{span}}({\bf H}_{11}{\bf V}\hat{\bf\Psi}_2^0\breve{\bf\Psi}_2^1)$, implies
\begin{align}
{\rm{span}}( {\bf{H}}_{12}{\bf W}\hat{\bf\Psi}_1^\prime\breve{\bf\Psi}_1^\prime)
\cap {\rm{span}}({\bf H}_{11}{\bf V}\hat{\bf\Psi}_2^0\breve{\bf\Psi}_2^1)={\bf 0}\label{eq36}.
\end{align}
Combining (\ref{eq35}) and (\ref{eq36}), we arrive at
\begin{align}
({\bf V}\hat{\bf\Psi}_2^0\breve{\bf\Psi}_2^1,
{\bf W}\hat{\bf\Psi}_1^\prime\breve{\bf\Psi}_1^\prime ) \in  \bar{\mathcal I}\label{eq37}.
\end{align}

Let ${\bf V}^\prime={\bf V}\hat{\bf\Psi}_2^0\breve{\bf\Psi}_2^1$ and
${\bf W}^\prime={\bf W}\hat{\bf\Psi}_1^\prime\breve{\bf\Psi}_1^\prime$. According to \emph{Corollary 2},
$d_s^1({\bf V}^\prime, {\bf W}^\prime)=
{\rm {rank}}\{{\bf{H}}_{11}{\bf V}\hat{\bf\Psi}_2^0\breve{\bf\Psi}_2^1\}$, which
together with (\ref{eq34}), gives $d_s^1({\bf V},{\bf W}) \le d_s^1({\bf V}^\prime, {\bf W}^\prime)$.
Besides, ${\rm{span}}({\bf H}_{21}{\bf V}^\prime) \subset {\rm{span}}({\bf H}_{21}{\bf V})$
and ${\rm{span}}({\bf H}_{22}{\bf W}^\prime) = {\rm{span}}({\bf H}_{22}{\bf W})$. So
$d_s^2({\bf V},{\bf W}) \le d_s^2({\bf V}^\prime, {\bf W}^\prime)$. This completes the proof.

\section{Proof of \emph{Corollary 1}} \label{appC}
Since by definition $\hat{\mathcal I} \subset \bar{\mathcal I}$, it holds that $\hat{\mathcal{D}} \subset \bar{\mathcal{D}}$.
In the sequel, we will show that for any given $({\bf V},{\bf W} )\in \bar{\mathcal I}$, we can
always construct another feasible point $({\bf V}^\star,{\bf W}^\star) \in \hat{\mathcal I}$,
which satisfy $d_s^1({\bf V}^\star,{\bf W}^\star)= d_s^1({\bf V} ,{\bf W} )$
and $d_s^2({\bf V}^\star,{\bf W}^\star)= d_s^2({\bf V} ,{\bf W} )$, thus giving the proof of
$\hat{\mathcal{D}} \supset \bar{\mathcal{D}}$. Concluding, it holds that $\bar{\mathcal{D}}=\hat{\mathcal{D}}$.

For any given $({\bf V} ,{\bf W})\in \bar{\mathcal I}$, ${\bf V} \in {\mathbb {C}}^{N_s^1 \times K_v}$,
${\bf W} \in {\mathbb {C}}^{N_s^2 \times K_w}$, we should have $({\bf V} ,{\bf W})\in \bar{\mathcal I}_1$
and $({\bf V} ,{\bf W})\in \bar{\mathcal I}_2$.
Since all channel matrices are assumed to be full rank,
it holds that ${\rm {rank}}\{{\bf{G}}_{2}{\bf W}\}=\min\{K_w,N_e\}$.

In the following, we consider two distinct cases.

(i) For the case of $K_w\ge N_e$, it holds that ${\rm {rank}}\{{\bf{G}}_{2}{\bf W}\}=N_e$.
Denote ${\bf{G}}_{2}{\bf W}=\left[{\bf U}_{1}\ {\bf U}_{0} \right]
\left[ {\begin{array}{*{20}{c}}
{{{\bf{\Sigma}}_{1}}}&{\bf{0}}\\
{\bf{0}}&{{{\bf{0}}}}
\end{array}} \right]
\left[ {\begin{array}{*{20}{c}}
{{{\bf{T}}_{1}^H}}\\
{{{\bf{T}}_{0}^H}}
\end{array}} \right]$ as the SVD of ${\bf{G}}_{2}{\bf W}$.
Then, the matrix ${\bf{G}}_{2}{\bf W}{\bf{T}}_{1}$ is invertible.
Let ${\bf B} ={\bf{T}}_{1}({\bf{G}}_{2}{\bf{W}}{\bf{T}}_{1} )^{-1}{\bf{G}}_{1}{\bf{V}} $.
Then,
\begin{align}
{\bf{G}}_{1}{\bf{V}} ={\bf{G}}_{2}{\bf{W}}{\bf{T}}_{1}({\bf{G}}_{2}{\bf{W}}{\bf{T}}_{1} )^{-1}{\bf{G}}_{1}{\bf{V}}= {\bf{G}}_{2}{\bf{W}}{\bf B}. \label{eq405}
\end{align}

(ii) For the case of $K_w< N_e$, ${\bf{G}}_{2}{\bf W}$ is full column rank.
Let ${\bf P}$ be the projection matrix of ${\bf{G}}_{2}{\bf{W}}$, i.e.,
\begin{align}
{\bf P}={\bf{G}}_{2}{\bf{W}} (({\bf{G}}_{2}{\bf{W}} )^H{\bf{G}}_{2}{\bf{W}} )^{-1}({\bf{G}}_{2}{\bf{W}} )^H.
\label{eq39}
\end{align}
Due to $({\bf V} ,{\bf W})\in \bar{\mathcal I}_1$, it holds that
\begin{align}
{\bf{G}}_{1}{\bf{V}}  &={\bf P}{\bf{G}}_{1}{\bf{V}}.   \label{eq40}
\end{align}
Substituting (\ref{eq39}) into (\ref{eq40}) and letting ${\bf B} =(({\bf{G}}_{2}{\bf{W}} )^H{\bf{G}}_{2}{\bf{W}} )^{-1}
({\bf{G}}_{2}{\bf{W}} )^H{\bf{G}}_{1}{\bf{V}} $, we arrive at
\begin{align}
{\bf{G}}_{1}{\bf{V}} ={\bf{G}}_{2}{\bf{W}} {\bf B}. \label{eq41}
\end{align}

Let ${\bf V}^\star={\bf V} $ and ${\bf W}^\star={\bf W}[{\bf B}\ {\bf B}^\perp]$.
Summarizing the above two cases, for both cases it holds that 
\begin{align}
&{\bf{G}}_{1}{\bf{V}}^\star={\bf{G}}_{2}{\bf{W}}^\star(:, 1:K_v), \nonumber 
\end{align}
which, combined with $({\bf V} ,{\bf W})\in \bar{\mathcal I}_2$, implies that
$({\bf V}^\star,{\bf W}^\star) \in \hat{\mathcal I}$.
On the other hand, since 
the matrix $[{\bf B}\  {\bf B}^\perp]$ is invertible, it holds that
$d_s^1({\bf V}^\star,{\bf W}^\star)= d_s^1({\bf V} ,{\bf W} )$
and $d_s^2({\bf V}^\star,{\bf W}^\star)= d_s^2({\bf V} ,{\bf W} )$.
This completes the proof.

\bibliography{mybib}

\begin{thebibliography}{10}
\providecommand{\url}[1]{#1}
\csname url@samestyle\endcsname
\providecommand{\newblock}{\relax}
\providecommand{\bibinfo}[2]{#2}
\providecommand{\BIBentrySTDinterwordspacing}{\spaceskip=0pt\relax}
\providecommand{\BIBentryALTinterwordstretchfactor}{4}
\providecommand{\BIBentryALTinterwordspacing}{\spaceskip=\fontdimen2\font plus
\BIBentryALTinterwordstretchfactor\fontdimen3\font minus
  \fontdimen4\font\relax}
\providecommand{\BIBforeignlanguage}[2]{{%
\expandafter\ifx\csname l@#1\endcsname\relax
\typeout{** WARNING: IEEEtran.bst: No hyphenation pattern has been}%
\typeout{** loaded for the language `#1'. Using the pattern for}%
\typeout{** the default language instead.}%
\else
\language=\csname l@#1\endcsname
\fi
#2}}
\providecommand{\BIBdecl}{\relax}
\BIBdecl

\bibitem{Wyner75}
A.~D. Wyner, ``The wire-tap channel,'' \emph{Bell Syst. Tech. J.}, vol.~54,
  no.~8, pp. 1355--1387, Jan. 1975.

\bibitem{Leung78}
S.~K. Leung-Yan-Cheong and M.~E. Hellman, ``The {Gaussian} wire-tap channel,''
  \emph{IEEE Trans. Inf. Theory}, vol.~24, no.~4, pp. 451--456, Jul. 1978.

\bibitem{Swindlehurst11}
S.~A.~A. Fakoorian and A.~L. Swindlehurst, ``Solutions for the {MIMO}
  {G}aussian wiretap channel with a cooperative jammer,'' \emph{IEEE Trans.
  Signal Process.}, vol.~59, no.~10, pp. 5013--5022, Oct. 2011.

\bibitem{Lingxiang14}
L.~Li, Z.~Chen, and J.~Fang, ``On secrecy capacity of {G}aussian wiretap
  channel aided by a cooperative jammer,'' \emph{IEEE Signal Process. Lett.},
  vol.~21, no.~11, pp. 1356--1360, Nov. 2014.

\bibitem{Han11}
H.-T. Chiang and J.~S. Lehnert, ``Optimal cooperative jamming for security,''
  in \emph{Proc. IEEE MILCOM}, Baltimore, MD, Nov. 2011, pp. 125--130.

\bibitem{Ali11}
S.~A.~A. Fakoorian and A.~L. Swindlehurst, ``Secrecy capacity of {MISO}
  {G}aussian wiretap channel with a cooperative jammer,'' in \emph{Proc. IEEE
  SPAWC}, San Francisco, CA, Jun. 2011, pp. 416--420.

\bibitem{Gan13}
G.~Zheng, I.~Krikidis, J.~Li, A.~P. Petropulu, and B.~Ottersten, ``Improving
  physical layer secrecy using full-duplex jamming receivers,'' \emph{IEEE
  Trans. Signal Process.}, vol.~61, no.~20, pp. 4962--4974, Oct. 2013.

\bibitem{zheng151}
Z.~Chu, K.~Cumanan, Z.~Ding, M.~Johnston, and S.~Y. Goff, ``Secrecy rate
  optimizations for a {MIMO} secrecy channel with a cooperative jammer,''
  \emph{IEEE Trans. Veh. Technol.}, vol.~64, no.~5, pp. 1833--1847, May 2015.

\bibitem{Zheng11}
G.~Zheng, L.-C. Choo, and K.-K. Wong, ``Optimal cooperative jamming to enhance
  physical layer security using relays,'' \emph{IEEE Trans. Signal Process.},
  vol.~59, no.~3, pp. 1317--1322, Mar. 2011.

\bibitem{Jiangyuan11}
J.~Li, A.~P. Petropulu, and S.~Weber, ``On cooperative relaying schemes for
  wireless physical layer security,'' \emph{IEEE Trans. Signal Process.},
  vol.~59, no.~10, pp. 4985--4997, Oct. 2011.

\bibitem{LunDong10}
L.~Dong, Z.~Han, A.~P. Petropulu, and H.~V. Poor, ``Improving wireless physical
  layer security via cooperating relays,'' \emph{IEEE Trans. Signal Process.},
  vol.~58, no.~3, pp. 1875--1888, Mar. 2010.

\bibitem{Shuangyu13}
S.~Luo, J.~Li, and A.~P. Petropulu, ``Uncoordinated cooperative jamming for
  secret communications,'' \emph{IEEE Trans. Inf. Forens. Security}, vol.~8,
  no.~7, pp. 1081--1090, Jul. 2013.

\bibitem{Kalogerias13}
D.~S. Kalogerias, N.~Chatzipanagiotis, M.~M. Zavlanos, and A.~P. Petropulu,
  ``Mobile jammers for secrecy rate maximization in cooperative networks,'' in
  \emph{Proc. IEEE ICASSP}, Vancouver, Canada, May 2013, pp. 2901--2905.

\bibitem{Wang09}
J.~Wang and A.~Swindlehurst, ``Cooperative jamming in {MIMO} ad hoc networks,''
  in \emph{Proc. Asilomar Conf. Signals, Syst. Comput.}, Pacific Grove, CA,
  Nov. 2009, pp. 1719--1723.

\bibitem{Hoon14}
J.~H. Lee and W.~Choi, ``Multiuser diversity for secrecy communications using
  opportunistic jammer selection: secure {DoF} and jammer scaling law,''
  \emph{IEEE Trans. Signal Process.}, vol.~62, no.~4, pp. 828--839, Feb. 2014.

\bibitem{Tao10}
J.~Zhu, J.~Mo, and M.~Tao, ``Cooperative secret communication with artificial
  noise in symmetric interference channel,'' vol.~14, no.~10, pp. 885--887,
  Oct. 2010.

\bibitem{Ali112}
S.~A.~A. Fakoorian and A.~L. Swindlehurst, ``{MIMO} interference channel with
  confidential messages: Achievable secrecy rates and precoder design,''
  \emph{IEEE Trans. Inf. Forens. Security}, vol.~6, no.~3, pp. 640--649, Sep.
  2011.

\bibitem{Koyluoglu11}
O.~O. Koyluoglu and H.~E. Gamal, ``Cooperative encoding for secrecy in
  interference channels,'' \emph{IEEE Trans. Inf. Theory}, vol.~57, no.~9, pp.
  5682--5694, Sep. 2011.

\bibitem{Xie15}
J.~Xie and S.~Ulukus, ``Secure degrees of freedom of {K-User Gaussian}
  interference channels: A unified view,'' \emph{IEEE Trans. Inf. Theory},
  vol.~61, no.~5, pp. 2647--2661, May 2015.

\bibitem{Xie142}
------, ``Secure degrees of freedom region of the {G}aussian interference
  channel with secrecy constraints,'' in \emph{Proc. IEEE ITW}, Hobart,
  Tasmania, Australia, Nov. 2014, pp. 361--365.

\bibitem{Koyluoglu112}
O.~O. Koyluoglu, H.~E. Gamal, L.~Lai, and H.~V. Poor, ``Interference alignment
  for secrecy,'' \emph{IEEE Trans. Inf. Theory}, vol.~57, no.~6, pp.
  3323--3332, Jun. 2011.

\bibitem{Tung15}
T.~T. Vu, H.~H. Kha, T.~Q. Duong, and N.-S. Vo, ``On the interference alignment
  designs for secure multiuser {MIMO} systems,'' \emph{[online], Available:
  http://arxiv.org/abs/1508.00349}.

\bibitem{Kalantari15}
A.~Kalantari, S.~Maleki, G.~Zheng, S.~Chatzinotas, and B.~Ottersten, ``Joint
  power control in wiretap interference channels,'' \emph{IEEE Trans. Wireless
  Commun.}, vol.~14, no.~7, pp. 3810--3823, Jul. 2015.

\bibitem{Lv15}
T.~Lv, H.~Gao, and S.~Yang, ``Secrecy transmit beamforming for heterogeneous
  networks,'' \emph{IEEE J. Sel. Areas Commun.}, vol.~33, no.~6, pp.
  1154--1170, Jun. 2015.

\bibitem{Nafea13}
M.~Nafea and A.~Yener, ``How many antennas does a cooperative jammer need for
  achieving the degrees of freedom of multiple antenna {G}aussian channels in
  the presence of an eavesdropper?'' in \emph{Proc. Allerton Conf.}, Allerton
  House, UIUC, Illinois, USA, Oct. 2013, pp. 774--780.

\bibitem{Nafea14}
------, ``Secure degrees of freedom for the {MIMO} wiretap channel with a
  multiantenna cooperative jammer,'' in \emph{Proc. IEEE ITW}, Hobart,
  Australia, Nov. 2014, pp. 626--630.

\bibitem{Nafea15}
------, ``Secure degrees of freedom of {$N$}-{$N$}-{$M$} wiretap channel with a
  {$K$}-antenna cooperative jammer,'' in \emph{Proc. IEEE ICC}, London, United
  Kingdom, Jun. 2015, pp. 4169--4174.

\bibitem{Agustin11}
A.~Agustin and J.~Vidal, ``Improved interference alignment precoding for the
  {MIMO $X$} channel,'' in \emph{Proc. IEEE ICC}, Kyoto, Japan, Jun. 2011, pp.
  1--5.

\bibitem{GouJafar10}
T.~Gou and S.~A. Jafar, ``Degrees of freedom of the {$K$}-user {$M\times N$}
  {MIMO} interference channel,'' \emph{IEEE Trans. Inf. Theory}, vol.~56,
  no.~12, pp. 6040--6057, Dec. 2010.

\bibitem{YetisGou10}
C.~M. Yetis, T.~Gou, S.~A. Jafar, and A.~H. Kayran, ``On feasibility of
  interference alignment in {MIMO} interference networks,'' \emph{IEEE Trans.
  Signal Process.}, vol.~58, no.~9, pp. 4771--4782, Sep. 2010.

\bibitem{Jiayi14}
J.~Chen, Q.~T. Zhang, and G.~Chen, ``Joint space decomposition-and-synthesis
  approach and achievable {DoF} regions for {$K$}-user {MIMO} interference
  channels,'' \emph{IEEE Trans. Signal Process.}, vol.~62, no.~9, pp.
  2304--2316, May 2014.

\bibitem{Paige81}
C.~Paige and M.~A. Saunders, ``Towards a generalized singular value
  decomposition,'' \emph{SIAM J. Numer. Anal.}, vol.~18, no.~3, pp. 398--405,
  Jun. 1981.

\bibitem{Liu09}
T.~Liu and S.~{Shamai (Shitz)}, ``A note on the secrecy capacity of the
  multi-antenna wire-tap channel,'' \emph{IEEE Trans. Inf. Theory}, vol.~55,
  no.~6, pp. 2547--2553, Jun. 2009.

\bibitem{Liu10}
R.~Liu, T.~Liu, H.~V. Poor, and S.~{Shamai (Shitz)}, ``Multiple-input
  multiple-output {G}aussian broadcast channels with confidential messages,''
  \emph{IEEE Trans. Inf. Theory}, vol.~56, no.~9, pp. 4215--4227, Sep. 2010.

\bibitem{OggierBabak11}
F.~Oggier and B.~Hassibi, ``The secrecy capacity of the {MIMO} wiretap
  channel,'' \emph{IEEE Trans. Inf. Theory}, vol.~57, no.~8, pp. 4961--4971,
  Aug. 2011.

\bibitem{Wornell11}
A.~Khisti and G.~Wornell, ``Secure transmission with multiple antennas-{II}:
  the {MIMOME} wiretap channel,'' \emph{IEEE Trans. Inf. Theory}, vol.~56,
  no.~11, pp. 5515--5532, Nov. 2010.

\end{thebibliography}
\bibliographystyle{IEEEtran}

\end{document}